\documentclass{ifacconf}

\usepackage{graphicx}      
\usepackage{natbib}        
\usepackage{amsmath}
\usepackage{amsfonts}
\usepackage{upgreek}
\usepackage{enumitem}
\usepackage{bm}
\usepackage{dsfont}
\usepackage{extpfeil}
\usepackage{mathtools}
\mathtoolsset{showonlyrefs}

\newcommand{\sgn}[2]{\left\lceil#1\right\rfloor^{#2}}

\newcommand{\mf}{\mathbf}
\usepackage[dvipsnames]{xcolor}

\begin{document}
\begin{frontmatter}

\title{Super-twisting over networks: A Lyapunov approach for distributed differentiation} 

\thanks[footnoteinfo]{
\tiny{This work was supported via projects PID2021-124137OB-I00 and 
PID2024-159279OB-I00 funded by MICIU/AEI/10.13039/501100011033 and by ERDF/EU and via 
project REMAIN S1/1.1/E0111 (Interreg Sudoe Programme, ERDF). Grant reference BG24/00121 funded by MICIU/AEI/10.13039/501100011033.
This work as also funded by the Gobierno de Aragón under Project DGA T45 23R, 
and by Spanish grant FPU20/03134.
{\textcolor{red}{This is a preprint of a manuscript submitted for possible publication. Copyright may be transferred without notice, after which this version may no longer be accessible.}}
}}

\author[First]{Rodrigo Aldana-López} 
\author[First]{Irene Pérez-Salesa} 
\author[Second]{David Gomez-Gutierrez} 
\author[First]{Rosario Aragüés} 
\author[First]{Carlos Sagüés} 

\address[First]{Departamento de Informática e Ingeniería de Sistemas (DIIS) and Instituto de Investigación en Ingeniería de Aragón (I3A), 
\\ Universidad de Zaragoza, 50018 Zaragoza, Spain.
(e-mail: rodrigo.aldana.lopez@gmail.com, i.perez@unizar.es, raragues@unizar.es, csagues@unizar.es)}

\address[Second]{Tecnológico Nacional de México, Instituto Tecnológico José Mario Molina Pasquel y Henríquez, 
Cam. Arenero 1101, 45019 Zapopan, Jalisco, Mexico.
(e-mail: david.gomez.g@ieee.org)}

\begin{abstract}                
We study distributed differentiation, where agents in a networked system estimate the average of local time-varying signals and their derivatives under mild assumptions on the agents' signals and their first and second derivatives. Existing sliding-mode methods provide only local stability guarantees and lack systematic gain selection. By isolating the structural features shared with the super-twisting algorithm and encoding them into an abstract model, we construct a Lyapunov function enabling systematic gain design and proving global finite-time convergence to consensus for the distributed differentiator. {Building on this framework, we develop an event-triggered hybrid system implementation using time-varying and state dependent threshold rules and derive minimum inter-event time guarantees and accuracy bounds that quantify the trade-off between estimation accuracy and communication effort.}
\end{abstract}
\end{frontmatter}

\section{Introduction}

The distributed differentiation problem  concerns a network of agents with local time-varying signals that aim to collaboratively estimate the derivative of their average in a distributed fashion, even under some form of persistent variation of the signals. This task can be understood as an extension of the Dynamic Average Consensus (DAC) problem. Applications of DAC include collective estimation in sensor networks, coordination in power grids, and formation control of multi-agent systems \citep{Solmaz2017}. Structurally, DAC algorithms are useful for hierarchical observer-based distributed coordination algorithms \citep{medcho, nhsObs}. One of the most popular applications of DAC nowadays is in distributed optimization \citep{guido2025}, where a DAC block is used for distributed gradient tracking. Recent advances in DAC include acceleration techniques \citep{edusebas2023} and resiliency mechanisms for unreliable communication \citep{iqbal2022}. Extending DAC to distributed differentiation is natural, since many applications require not only the average signal but also its derivatives, for instance, {to design distributed trajectory tracking controllers, sensor fusion, or to estimate higher-order behaviors of the network \citep{pierri2022,medcho,aldana2023a}}. Nevertheless, most DAC approaches fail to converge to the true derivative under persistently varying signals, achieving only ultimately bounded stability due to the inherent limitations of their linear structure. To address this issue, \citep{edcho} introduced the concept of Exact Dynamic Consensus (EDC) based on High-Order Sliding-Mode (HOSM) techniques. The first application of sliding modes in this context relied on First-Order Sliding Modes (FOSM)~\citep{freeman2019}. However, FOSM schemes are prone to the well-known chattering phenomenon \citep[Ch.~3]{fridman2002}, consisting of high-frequency oscillations that render implementations sensitive to noise and delays. Additionally, the approach by~\cite{freeman2019} recovers only the average signal but not its derivative. HOSM provide a principled alternative, mitigating chattering while enabling exact computation of higher-order derivatives when applied to observer design. The Exact Dynamic Consensus of High-Order (EDCHO) algorithm \citep{edcho} leverages these properties to provide the first distributed differentiator in the literature, later extended to a Robust EDCHO (REDCHO) in \citep{redcho}. A complete distributed differentiator was presented in \citep{aldana2023}, including a noise robustness analysis that does not require external local differentiators. 
Despite their versatility, current formulations of EDC-based distributed differentiators lack systematic parameter tuning guidelines, leaving practitioners to rely on simulations for gain selection. Moreover, for REDCHO only local stability has been established, even though simulations suggest that global stability should also be provable. 

The stability analysis of these algorithms is closely tied to the theory of HOSM differentiators~\citep{levant2003}, particularly the super-twisting algorithm \citep{Levant1998}, which is the prototypical case for a first-order differentiator. For this reason, we briefly review existing stability proof approaches for super-twisting and identify the elements that can be carried over to the distributed differentiator setting. Early analyses exploited its two-dimensional structure via majorant curves, which do not generalize to higher-order or multi-agent settings. Homogeneity-based approaches were later developed \citep{levant2003,levant2005} and used in the original analysis of EDCHO. However, they do not yield explicit gain conditions. More principled Lyapunov analyses of super-twisting appeared in \citep{yuri2010,moreno2012}, later refined to obtain explicit feasible gains and settling-time estimates in \citep{seeber2017,seeber2018}, with extensions to arbitrary-order algorithms in \citep{Cruz2019,seeber2023}. Structural approaches have also been proposed: a suitable change of variables can transform super-twisting into a linear system scaled by a nonlinear scalar function, from which a Lyapunov function is inferred. This idea was applied in \citep{hernan2019} to generalized super-twisting–like systems and in \citep{seeber2021} to predefined-time differentiators. However, such a transformation is not directly applicable in the multi-agent setting, where coupling prevents a unified change of variables. {Recent works in \citep{geromel2026, moreno2021multivariable,lopez2019generalised} introduce generalized multi-variable super-twisting algorithms with constructive Lyapunov analysis. However, they do not accommodate the graph-dependent coupling required in multi-agent settings.}

In this work, we develop a Lyapunov framework for the first-order distributed differentiator. We extract the minimal structural features shared with super-twisting and formalize them into an abstract super-twisting system. By constructing a Lyapunov function for this abstract system, we show that the stability proofs for both the super-twisting algorithm and the distributed differentiator are particular instances. This unifying perspective yields explicit gain design conditions and extends the theoretical guarantees of REDCHO beyond local stability. {We show that this framework is advantageous for the computation of accuracy bounds in alternative event-triggered implementations.} 

\subsection{Notation}
\label{sec:notation}
Let $\mathds{1}=[1,\dots,1]^\top$ denote the vector of ones of appropriate dimension.  
Define $\text{sign}(x)=1$ if $x>0$, $\text{sign}(0)=0$, and $\text{sign}(x)=-1$ if $x<0$.  
For $x\in\mathbb{R}$ and $\alpha>0$, let $\lceil x\rfloor^\alpha := |x|^\alpha \text{sign}(x)$, and set $\lceil x\rfloor^0 := \text{sign}(x)$.  
For a vector $\mf{x}=[x_1,\dots,x_n]^\top \in \mathbb{R}^n$, define
$
\sgn{\mf{x}}{\alpha} := \begin{bmatrix}
\sgn{x_1}{\alpha},\dots,\sgn{x_n}{\alpha}
\end{bmatrix}^\top
$
with $\alpha \ge 0$.
For $\mf{r}=[r_1,\dots,r_n]^\top$, define the weighted norm
$
\|\mf{x}\|_\mf{r} := \sum_{i=1}^n |x_i|^{1/r_i}.
$
For two vectors $\mf{x}, \mf{y}$, let their inner product be written as $\langle \mf{x}, \mf{y}\rangle := \mf{x}^\top\mf{y}$.

\section{Distributed Differentiation}
\label{sec:distributed:differentiation}
We formulate the problem of first-order distributed differentiation as follows.  
Consider $N \in \mathbb{N}$ agents in a networked system interacting over a communication network represented by an undirected graph $\mathcal{G}=(\mathcal{V},\mathcal{E})$, where $\mathcal{V}=\{1,\dots,N\}$ is the set of nodes and $\mathcal{E}\subseteq\mathcal{V}\times\mathcal{V}$ is the set of edges.  
An edge $(i,j)\in\mathcal{E}$ represents a bidirectional communication link between agents $i$ and $j$.  
For each agent $i\in\mathcal{V}$, let $\mathcal{N}_i\subseteq\mathcal{V}$ denote its set of neighbors.

Each agent $i\in\mathcal{V}$ has access to a local time-varying signal $s_i(t)$, {assumed to be twice differentiable.} 
The collective objective is for all agents to recover the global average
$\bar{s}(t) = \frac{1}{N}\sum_{i=1}^N s_i(t),$
together with its time derivative $\dot{\bar{s}}(t)$, using only one-hop local information exchange over $\mathcal{G}$. For the analysis, we work under the following boundedness condition.

\begin{assum}
\label{as:bound}
Given $\gamma\geq 0$, there exists a known $L\geq 0$ such that for each agent $i$ and all $t\geq 0$,
\[
\Big|\ddot{\bar s}(t) - \ddot{s}_i(t) 
+ 2\gamma \big(\dot{\bar s}(t) - \dot{s}_i(t)\big) 
+ \gamma^2 \big(\bar s(t) - s_i(t)\big)\Big| \leq \frac{L}{\sqrt{N}}.
\]
\end{assum}

{Assumption~\ref{as:bound} characterizes the level of mismatch between the local signals and their average that the distributed differentiation algorithm is able to tolerate. A practical way to reason about this condition is through local bounds on the signals and their derivatives. For instance, when $\gamma=0$, uniform bounds on the second derivatives $\ddot s_i(t)$ are sufficient to guarantee Assumption~\ref{as:bound}, recovering a condition of the same type commonly used in standard sliding-mode differentiators. For $\gamma>0$, boundedness of $s_i(t)$ and its derivatives can be used locally  to imply the assumption, which is satisfied in many applications e.g., where reference signals can be decomposed as sinusoids. It is worth noting that $N$ can be obtained in a distributed fashion in finite time using known methods~\citep{medcho}.
}

To this end, the REDCHO instance for first-order distributed differentiation protocol proposed in \cite{redcho} takes the form:

\begin{equation}
\label{eq:redcho}
\begin{aligned}
&\text{\bf Protocol dynamics:}\\[0.4em]
&\begin{array}{ll}
\dot{{\upeta}}_{i,0}(t)&= k_{0}\sqrt{L} \sum_{j\in\mathcal{N}_i}\sgn{\hat{s}_{i,0}(t) - \hat{s}_{j,0}(t)}{1/2} 
\!+\! {\upeta}_{i,1}(t)\! -\! \gamma {\upeta}_{i,0}(t)\\[0.8em]
\dot{{\upeta}}_{i,1}(t)&=k_1 L \sum_{j\in\mathcal{N}_i}\text{sign}\big(\hat{s}_{i,0}(t) - \hat{s}_{j,0}(t)\big) 
-\gamma {\upeta}_{i,1}(t)
\end{array}\\[0.6em]
&\text{\bf Shared information:}\\[0.6em]
&\hat{s}_{i,0}(t) = s_i(t) - {\upeta}_{i,0}(t),\\[0.6em]
&\text{\bf Differentiator output:}\\[0.6em]
&\hat{s}_{i,1}(t) = \dot{s}_i(t) - {\upeta}_{i,1}(t) + \gamma \upeta_{i,0}(t),
\end{aligned}
\end{equation}
where $k_0,k_1>0$, $\gamma\geq 0$ are design parameters, and $L$ is chosen from Assumption \ref{as:bound}.

{We briefly explain the intuition behind \eqref{eq:redcho}. The correction terms 
$\sgn{\hat{s}_{i,0}(t)-\hat{s}_{j,0}(t)}{1/2}$ and 
$\text{sign}(\hat{s}_{i,0}(t)-\hat{s}_{j,0}(t))$ are introduced to enforce consensus using only relative information. The fractional power term is selected to induce weighted homogeneity, a key structural property that enables the Lyapunov based stability analysis. The discontinuous sign term complements this mechanism by rejecting unknown bounded effects induced by the tracking of the average signal and its derivative.  Both correction terms depend exclusively on differences $\hat{s}_{i,0}(t)-\hat{s}_{j,0}(t)$, so agent $i$ only needs to communicate $\hat{s}_{i,0}(t)$. }

{The protocol further introduces the linear terms $-\gamma\upeta_{i,0}(t)$ and $-\gamma\upeta_{i,1}(t)$ to remove the classical initialization constraints of DAC schemes, which would otherwise require $\sum_{i=1}^N\upeta_{i,\mu}(0)=0, \mu=0,1$. When $\gamma>0$, any initial condition mismatch decays exponentially with rate $\gamma$.} {Finally, once consensus is achieved, $\upeta_{i,0}(t)$ converges to the exact correction required to cancel local offsets, yielding $\hat{s}_{i,0}(t)=\bar{s}(t)$. The derivative is then reconstructed via $\hat{s}_{i,1}(t)=\dot{s}_i(t)-\dot{\upeta}_{i,0}(t)$. Under convergence, the internal dynamics satisfy $\dot{\upeta}_{i,0}(t)=\upeta_{i,1}(t)-\gamma\upeta_{i,0}(t)$, which leads directly to the differentiator output definition in \eqref{eq:redcho}.
}  
\begin{rem}
\label{rem:no:diff}
{A modification of \eqref{eq:redcho} was proposed in \citep{aldana2023} to make the differentiation output $\hat{s}_{i,1}(t)$ not depend on the local derivative $\dot{s}_i(t)$. Such modified protocol can be written as:}
{
\begin{equation}
\label{eq:redcho:2}
\begin{aligned}
&\begin{array}{ll}
\dot{{\upeta}}'_{i,0}(t)
&= k_{0}\sqrt{L} \sum_{j\in\mathcal{N}_i}
\sgn{\hat{s}'_{i,0}(t) - \hat{s}'_{j,0}(t)}{1/2}
 - \gamma {\upeta}'_{i,0}(t)\\
&+ {\upeta}'_{i,1}(t) + 
+k_{0}\sqrt{L}\sgn{\hat{s}'_{i,0}(t) - \hat{s}_{i,0}(t)}{1/2} \\[0.8em]
\dot{{\upeta}}'_{i,1}(t)
&=k_1 L \sum_{j\in\mathcal{N}_i}
\text{sign}\big(\hat{s}'_{i,0}(t) - \hat{s}'_{j,0}(t)\big) 
-\gamma {\upeta}'_{i,1}(t) \\
&
+k_{0}\sqrt{L}\ \text{sign}(\hat{s}'_{i,0}(t) - \hat{s}_{i,0}(t)) \\[0.8em]
\dot{{\upeta}}_{i,0}(t)
&= k_{0}\sqrt{L}
\sgn{\hat{s}_{i,0}(t) - \hat{s}'_{i,0}(t)}{1/2}
+ {\upeta}_{i,1}(t) - \gamma {\upeta}_{i,0}(t)\\[0.8em]
\dot{{\upeta}}_{i,1}(t)
&=k_1 L \text{sign}\big(\hat{s}_{i,0}(t) - \hat{s}'_{i,0}(t)\big) 
-\gamma {\upeta}_{i,1}(t)
\end{array}
\end{aligned}
\end{equation}}
{where 
$$
\begin{aligned}
\hat{s}_{i,0}'(t) = s_i(t) - {\upeta}_{i,0}'(t),\quad
\hat{s}_{i,0}(t) =  - {\upeta}_{i,0}(t) 
\end{aligned}
$$
and $\hat{s}_{i,1}(t) = 2( \gamma \upeta_{i,0}(t)- {\upeta}_{i,1}'(t))$ is the distributed differentiator output. As discussed in \citep{redcho}, this protocol coincides exactly with REDCHO when applied to an augmented network in which virtual nodes are appended to each agent. Consequently, the convergence analysis developed for REDCHO applies directly and yields the same conclusions for \eqref{eq:redcho:2}.}
\end{rem}

The first main result of this work is given below, and characterizes convergence for \eqref{eq:redcho}.

\begin{thm}
\label{th:main}
Suppose Assumption~\ref{as:bound} holds for some $L,\gamma \ge 0$.  
Then, there exist gains $k_0,k_1>0$ such that, for any initial conditions 
${\upeta}_{i,0}(0),{\upeta}_{i,1}(0)\in\mathbb{R}$, the trajectories of \eqref{eq:redcho} satisfy
\[
\hat{s}_{i,0}(t) = \hat{s}_{j,0}(t), \quad 
\hat{s}_{i,1}(t) = \hat{s}_{j,1}(t),
\]
for all $(i,j)\in\mathcal{E}$ and all $t\ge T$, for some finite $T\ge 0$.  
Moreover, each agent asymptotically recovers the global signals,
\[
\lim_{t\to\infty} |\hat{s}_{i,0}(t)-\bar{s}(t)|=0, 
\qquad 
\lim_{t\to\infty} |\hat{s}_{i,1}(t)-\dot{\bar{s}}(t)|=0.
\]
\end{thm}

Theorem \ref{th:main} extends the local stability result of \cite{redcho} to global stability.  
It also strengthens \cite{edcho} by providing systematic design conditions for $k_0,k_1>0$ as a result of the proof.  

{To prove Theorem \ref{th:main}, we analyze its error system, developed as follows. We first rewrite \eqref{eq:redcho} in compact form as
$$
\begin{aligned}
\dot{\bm{\upeta}}_{0}(t) &= k_{0}\sqrt{L}\mf{D}\sgn{\mf{D}^\top \hat{\mf{s}}_{0}(t)}{\frac{1}{2}}
    + \bm{\upeta}_{1}(t) - \gamma \bm{\upeta}_{0}(t), \\[1em]
\dot{\bm{\upeta}}_{1}(t) &= k_{1}L\mf{D}\sgn{\mf{D}^\top \hat{\mf{s}}_{0}(t)}{0}
    - \gamma \bm{\upeta}_{1}(t),
\end{aligned}
$$
where the auxiliary variables are defined as
\begin{equation}
\label{eq:variables:hat:s}
\begin{aligned}
\hat{\mf{s}}_{0}(t) &= \mf{s}(t) - \bm{\upeta}_{0}(t), \\
\hat{\mf{s}}_{1}(t) &= \dot{\mf{s}}(t) + \gamma \mf{s}(t) - \bm{\upeta}_{1}(t),
\end{aligned}
\end{equation}
with $\bm{\upeta}_{\mu}(t) = [\upeta_{1,\mu}(t),\dots,\upeta_{N,\mu}(t)]^\top$, $\mu\in\{0,1\}$, 
$\mf{s}(t) = [s_{1}(t),\dots,s_{N}(t)]^\top$, and $\mf{D}$ denoting the $\{1,-1,0\}$ incidence matrix of $\mathcal{G}$. Define the consensus errors
\begin{equation}
\label{eq:variables:e}
\begin{aligned}
\mf{e}_{0}(t) = \mf{P}\hat{\mf{s}}_{0}(t), \quad \mf{e}_{1}(t) = \mf{P}\hat{\mf{s}}_{1}(t),
\end{aligned}
\end{equation}
with $\mf{P} = \mf{I} - \frac{1}{N}\mathds{1}\mathds{1}^\top$. Their dynamics are
\begin{equation}
\label{eq:error:system}
\begin{aligned}
\dot{\mf{e}}_{0}(t) &= -k_{0}\sqrt{L}\mf{D}\sgn{\mf{D}^\top \mf{e}_{0}(t)}{\frac{1}{2}}
    + \mf{e}_{1}(t) - \gamma \mf{e}_{0}(t), \\[1em]
\dot{\mf{e}}_{1}(t) &= -k_{1}L\mf{D}\sgn{\mf{D}^\top \mf{e}_{0}(t)}{0}
    - \gamma \mf{e}_{1}(t) + \mf{d}(t),
\end{aligned}
\end{equation}
where
$$
\mf{d}(t) = \mf{P}\big(\ddot{\mf{s}}(t) + 2\gamma \dot{\mf{s}}(t) + \gamma^{2}\mf{s}(t)\big).
$$
implying $\mf{d}(t)\in L\mathcal{D}$ due to Assumption \ref{as:bound} with:
$$
\begin{aligned}
\mathcal{D} &= \{\mf{d}\in\mathcal{X} : \|\mf{d}\|\leq 1\}.
\end{aligned}
$$
Convergence as stated in Theorem \ref{th:main} relies on the convergence of the dynamical system in \eqref{eq:error:system} towards the origin. Such a convergence proof is deferred to Section \ref{sec:proof}, as it requires the development of a novel Lyapunov-based approach, introduced in the next section, which we term the \emph{abstract super-twisting method}.
}

\section{Abstract Super-Twisting}
\label{sec:abstract:super:twisting}

In the following, we make extensive use of the concept of homogeneity, for which we refer the reader to the Appendix. To study \eqref{eq:error:system} we will consider an abstract version of it first. Let the abstract differential inclusion
\begin{equation}
\label{eq:abstract:sys}
    \begin{aligned}
\dot{\mf{e}}_0(t) &= -k_0 L^{1/2}\nabla U(\mf{e}_0(t)) + \mf{e}_1(t) - \gamma \mf{e}_0(t), \\[0.4em]
\dot{\mf{e}}_1(t) &\in -k_1 L \mathcal{S}(\mf{e}_0(t)) - \gamma \mf{e}_1(t) + L \mathcal{D},
    \end{aligned}
\end{equation}
with $\mathcal{D}\subset\mathcal{X}\subseteq\mathbb{R}^N$. The abstract super-twisting system in \eqref{eq:abstract:sys} is introduced as a prototypical model central to the analysis. The following assumptions are imposed and will later be verified on the error dynamics in \eqref{eq:error:system}, thereby linking the model directly to the proof of Theorem~\ref{th:main}.

\begin{assum}
\label{as:basic}
The following properties hold:
\begin{enumerate}[label=\roman*)]
    \item $\mathcal{X}\subseteq\mathbb{R}^N$ is a vector space.
    \item $U:\mathcal{X}\to\mathbb{R}$ is strictly convex, positive definite, differentiable everywhere, and homogeneous of degree $3/2$.
    \item There exists $c_\mathcal{S}>0$ such that $\mathcal{S}:\mathcal{X}\rightrightarrows \mathcal{X}$ satisfies
    \[
    \langle \mf{s}, \mf{e}_0 \rangle \geq c_\mathcal{S}\|\mf{e}_0\|, \quad \forall \mf{s}\in \mathcal{S}(\mf{e}_0),~ \forall \mf{e}_0\in \mathcal{X}.
    \]
    Moreover, $\mathcal{S}$ is homogeneous of degree $0$.
    \item $\mathcal{D}\subset \mathcal{X}$ is compact with $\sup_{\mf{d}\in\mathcal{D}}\|\mf{d}\|=1$.
    \item Let $\partial U:\mathcal{X}\rightrightarrows\mathcal{X}$ be the subdifferential defined by $\partial U(\mf{e}_0) = \{\nabla U(\mf{e}_0)\}$ for all $\mf{e}_0\in\mathcal{X}$.  
    Assume that $\partial U$ and $\mathcal{S}$ are nonempty, bounded, closed, convex, and upper semi-continuous.\footnote{These conditions ensure existence of solutions to \eqref{eq:abstract:sys} in the sense of Filippov \citep{filippov}, and are trivially satisfied in the examples considered in this work.}
\end{enumerate}
\end{assum}

\begin{rem}
\label{rem:assump:super:twist}
The system \eqref{eq:abstract:sys} is referred to here as the \emph{abstract super-twisting} system, since it reduces to the classical scalar super-twisting dynamics under the choice 
$N=1$, $\gamma=0$, $U(e_0)=(2/3)|e_0|^{3/2}$, $\mathcal{S}(e_0)=\text{sign}(e_0)$, $\mathcal{X}=\mathbb{R}$, and $\mathcal{D}=[-1,1]$.  
These selections satisfy Assumption~\ref{as:basic} with $c_\mathcal{S}=1$, yielding
\begin{equation}
\label{eq:st}
\begin{aligned}
\dot{e}_0(t) &= -k_0 L^{1/2}\sgn{e_0(t)}{1/2} + e_1(t), \\[0.4em]
\dot{e}_1(t) &\in -k_1 L \text{sign}(e_0(t)) + [-L,L],
\end{aligned}
\end{equation}
which coincides with the super-twisting system~\citep{Levant1998}.  
{This example makes it evident why the $3/2$ homogeneity degree is needed, so that its gradient generates the $1/2$ fractional power required by the super-twisting. Moreover, as we will discuss in detail in Section \ref{sec:proof}, Assumption \ref{as:basic} is not restrictive in the sense that it also covers the error system of interest in \eqref{eq:error:system}, by picking $\nabla U({\mf{e}}_0)=\mf{D}\sgn{\mf{D}^{\top}{\mf{e}}_0}{1/2}$ and $\mathcal{S}(\mf{e}_0) = \mf{D}\sgn{\mf{D}^{\top}{\mf{e}}_0}{0}$}. 
\end{rem}

\begin{rem}
\label{rem:mgsta}
{There are other generalizations of the super-twisting algorithm in the literature, such as those proposed in \citep{ moreno2021multivariable, lopez2019generalised}. These frameworks cannot be applied to the distributed differentiator associated with the error system \eqref{eq:error:system}. Indeed, to apply those generalizations, one would need the existence of a constant $c>0$ such that
\begin{equation}
\label{eq:mgsta}
\mf{D}\,\sgn{\mf{D}^\top \mf{e}_0}{0}
=
c\,\mf{J}(\mf{e}_0)\,\mf{D}\,\sgn{\mf{D}^\top \mf{e}_0}{1/2},
\end{equation}
for all $\mf{e}_0\in\mathbb{R}^N\setminus\text{span}(\mathds{1})$, where $\mf{J}(\mf{e}_0)$ denotes the Jacobian of $\mf{D}\,\sgn{\mf{D}^\top \mf{e}_0}{1/2}$, given by
$$
\mf{J}(\mf{e}_0)
=
\tfrac12\,\mf{D}\,\mathrm{diag}\!\big(|\mf{D}^\top \mf{e}_0|^{-1/2}\big)\,\mf{D}^\top
$$
where $\mathrm{diag}(|\mf{D}^\top \mf{e}_0|^{-1/2})$ denotes the diagonal matrix whose diagonal entries are the inverse square roots of the absolute values of the elements of the vector $\mf{D}^\top \mf{e}_0$.
For \eqref{eq:mgsta} to hold, it is necessary that $\mf{D}^\top\mf{D}=\alpha \mf{I}$ for some $\alpha\in\mathbb{R}$, which fails in any undirected connected graph except the trivial case of $N=2$. }
\end{rem}

We are now ready to state the main result regarding the abstract super-twisting.

\begin{thm}
\label{th:abstract:st}
Under Assumption~\ref{as:basic}, if the gains $k_0,k_1>0$ and $\gamma\ge 0$ satisfy
\begin{equation}
\label{eq:gains}
k_0 > \sup_{\|\mf{x}\|_\mf{r}=1, \Gamma(\mf{x})>0} 
\frac{\Pi(\mf{x})}{\Gamma(\mf{x})}, 
\qquad 
k_1 > \frac{1}{c_\mathcal{S}},
\end{equation}
with $\Gamma$ and $\Pi$ defined in \eqref{eq:functions}, then the origin of \eqref{eq:abstract:sys} is finite-time stable for all initial conditions $\mf{e}_0(0),\mf{e}_1(0)\in\mathcal{X}$.
\end{thm}

{The main outcome of the previous result is that it will allow us to analyze stability of the error system \eqref{eq:error:system} for the distributed differentiator. The rest of this section is devoted to the proof of Theorem~\ref{th:abstract:st}.}

\subsection{Invariance and Lyapunov function candidate}
\label{sec:invariance}
Throughout this section, time-dependence will be omitted to emphasize algebraic relations in the Lyapunov analysis. 

{The proof strategy is the following. First, in Lemma \ref{lem:change} we use the change of variables
\begin{equation}
\label{eq:variables:x}
\mf{x}_0 = \frac{\mf{e}_0}{L}, 
\qquad 
\mf{x}_1 = \frac{\mf{e}_1}{k_0 L},
\end{equation}
to rewrite the abstract super-twisting dynamics in a form that is independent of the constant $L$, making explicit that variations in $L$ correspond to a rescaling of the system gains. Next, Lemma \ref{lem:invariance} establishes that the vector space $\mathcal{X}$ is invariant under the closed loop dynamics. This property is essential, since several subsequent arguments and conclusions are only valid when the trajectories evolve within $\mathcal{X}$. The core of the analysis is based on the \emph{Lyapunov function candidate}
\begin{equation}
\label{eq:lyapunov}
V(\mf{x}_0, \mf{x}_1)
\;=\;
U(\mf{x}_0) + (1+\beta) U^*(\mf{x}_1) - \langle \mf{x}_0, \mf{x}_1 \rangle,
\end{equation}
with $\beta>0$, where
\begin{equation}
\label{eq:dual}
U^*(\mf{x}_1)
\;=\;
\sup_{\mf{x}_0 \in \mathcal{X}}
\big\{ \langle \mf{x}_0, \mf{x}_1 \rangle - U(\mf{x}_0) \big\}
\end{equation}
denotes the convex conjugate of $U$ \cite[Section 3.3.1]{boyd_convex}. Lemma \ref{le:gamma} shows that this candidate is indeed a valid Lyapunov function for the system. Finally, Section \ref{sec:lyap:convergence} shows that the proposed Lyapunov function complies with a differential inequality implying finite time convergence.}

{Following this strategy, we begin by introducing a change of variables to rewrite \eqref{eq:abstract:sys} in a more convenient form.}

\begin{lem}
\label{lem:change}
Let Assumption~\ref{as:basic} hold.  
With the change of variables \eqref{eq:variables:x}, system~\eqref{eq:abstract:sys} is equivalently written as
\begin{equation}
\label{eq:abstract:sys:2}
\begin{aligned}
\dot{\mf{x}}_0 &= -\gamma \mf{x}_0 - \tilde{k}_0\big(\nabla U(\mf{x}_0) - \mf{x}_1\big), \\[0.4em]
\dot{\mf{x}}_1 &\in -\gamma \mf{x}_1 - \tilde{k}_1\left(\mathcal{S}(\mf{x}_0) + \frac{1}{k_1}\mathcal{D}\right),
\end{aligned}
\end{equation}
where $
\tilde{k}_0 = k_0, 
\tilde{k}_1 = {k_1}/{k_0}.$
\end{lem}

\begin{pf}
We compute directly:
\[
\begin{aligned}
\dot{\mf{x}}_0 &= L^{-1}\dot{\mf{e}}_0 \\
&= L^{-1}\big(-k_0 L^{1/2}\nabla U(\mf{e}_0) + \mf{e}_1 - \gamma \mf{e}_0\big) \\
&= L^{-1}\big(-k_0 L^{1/2}\nabla U(L\mf{x}_0) + k_0L \mf{x}_1 - \gamma L \mf{x}_0\big) \\
&= -\gamma \mf{x}_0 - k_0\big(\nabla U(\mf{x}_0) - \mf{x}_1\big),
\end{aligned}
\]
since $\nabla U(\bullet)$ is homogeneous of degree $1/2$, so $\nabla U(L\mf{x}_0)=L^{1/2}\nabla U(\mf{x}_0)$.  
Similarly,
\[
\begin{aligned}
\dot{\mf{x}}_1 &= (k_0L)^{-1}\dot{\mf{e}}_1 \\
&\in (k_0L)^{-1}\big(-k_1L \mathcal{S}(\mf{e}_0) - \gamma \mf{e}_1 + L \mathcal{D}\big) \\
&= (k_0L)^{-1}\big(-k_1L \mathcal{S}(L\mf{x}_0) - (k_0L)\gamma \mf{x}_1 + L \mathcal{D}\big) \\
&= -\gamma \mf{x}_1 - \frac{k_1}{k_0}\left(\mathcal{S}(\mf{x}_0) + \frac{1}{k_1}\mathcal{D}\right),
\end{aligned}
\]
since $\mathcal{S}(\bullet)$ is homogeneous of degree $0$, so $\mathcal{S}(L\mf{x}_0)=\mathcal{S}(\mf{x}_0)$.  
\qed
\end{pf}

{Now that the system has been written in the form of \eqref{eq:abstract:sys:2}, the following lemma concludes that its trajectories evolve within the space $\mathcal{X}$.}

\begin{lem}
\label{lem:invariance}
Let Assumption~\ref{as:basic} hold.  
Then $\mathcal{X}$ is forward invariant for solutions  $\mf{x}_0(t),\mf{x}_1(t)$ of \eqref{eq:abstract:sys:2}.
\end{lem}

\begin{pf}
Observe that 
\[
\tilde{k}_1\Big(\mathcal{S}(\mf{x}_0) + \frac{1}{k_1}\mathcal{D}\Big)\subseteq \mathcal{X},
\]
since $\mathcal{X}$ is closed under addition. Thus, the right-hand side of \eqref{eq:abstract:sys:2} belongs to $\mathcal{X}$ whenever the state does.  
Hence, starting from $\mf{x}_0(0),\mf{x}_1(0)\in\mathcal{X}$, invariance follows directly by integration.  
\qed
\end{pf}

{Next, we proceed to analyze the properties of the proposed Lyapunov function in \eqref{eq:lyapunov}.}

\begin{lem}
\label{le:gamma}
Let Assumption~\ref{as:basic} hold.  
Then, the Lyapunov function $V$ in \eqref{eq:lyapunov} is positive definite and radially unbounded.
\end{lem}

\begin{pf}
From \eqref{eq:lyapunov},
\[
V(\mf{x}_0,\mf{x}_1) \;=\; W(\mf{x}_0,\mf{x}_1) + \beta U^*(\mf{x}_1),
\]
with $W(\mf{x}_0,\mf{x}_1) := U(\mf{x}_0) + U^*(\mf{x}_1) - \langle \mf{x}_0,\mf{x}_1\rangle$.   In the following, we make use of the properties from Lemma~\ref{le:properties} in the Appendix. First, $W(\mf{x}_0,\mf{x}_1)\ge 0$ with equality if and only if $\mf{x}_1=\nabla U(\mf{x}_0)$. 
At the same time, $\beta U^*(\mf{x}_1)\ge 0$, and equality holds only when $\mf{x}_1=\mf{0}$.  
Thus $V(\mf{x}_0,\mf{x}_1)=0$ implies $\mf{x}_1=\mf{0}$ and $\mf{x}_0=\mf{0}$, so $V$ is positive definite. Moreover, since $U$ is homogeneous of degree $3/2$, its convex dual $U^*$ is homogeneous of degree $3$.  
Consequently $V$ is $\mf{r}$-homogeneous of degree $3$ with $\mf{r} = [2\mathds{1}^\top, \mathds{1}^\top]$, which implies radial unboundedness.  
\qed
\end{pf}

\subsection{Convergence of the abstract super-twisting}
\label{sec:lyap:convergence}
We prove Theorem~\ref{th:abstract:st} using the Lyapunov function \eqref{eq:lyapunov} with $\beta\geq 7$ as in Lemma \ref{le:properties}-\ref{prop:gamma} in the Appendix.  
Along the trajectories of \eqref{eq:abstract:sys:2}, the derivative of $V$ is
\[
\begin{aligned}
&\dot{V}=\\ 
&\langle \nabla U(\mf{x}_0), \dot{\mf{x}}_0\rangle 
   + (1+\beta)\langle \nabla U^*(\mf{x}_1), \dot{\mf{x}}_1\rangle 
   - \langle \mf{x}_1, \dot{\mf{x}}_0\rangle 
   - \langle \mf{x}_0, \dot{\mf{x}}_1\rangle \\
&= \langle \nabla U(\mf{x}_0) - \mf{x}_1, \dot{\mf{x}}_0\rangle
   + \langle (1+\beta)\nabla U^*(\mf{x}_1) - \mf{x}_0, \dot{\mf{x}}_1\rangle.
\end{aligned}
\]
Substituting the dynamics \eqref{eq:abstract:sys:2} gives
\[
\begin{aligned}
&\dot{V} \in 
\langle \nabla U(\mf{x}_0) - \mf{x}_1,\; -\gamma \mf{x}_0 - \tilde{k}_0(\nabla U(\mf{x}_0) - \mf{x}_1)\rangle \\[0.3em]
&+ \langle (1+\beta)\nabla U^*(\mf{x}_1) - \mf{x}_0,\; -\gamma \mf{x}_1 - \tilde{k}_1\big(\mathcal{S}(\mf{x}_0) + \mf{d}/k_1\big)\rangle,
\end{aligned}
\]
for some $\mf{d}\in\mathcal{D}$.  
{Rearranging terms yields}
\[
\begin{aligned}
&\dot{V} \in 
\langle \nabla U(\mf{x}_0) - \mf{x}_1, -\gamma \mf{x}_0\rangle 
+ \langle (1+\beta)\nabla U^*(\mf{x}_1) - \mf{x}_0, -\gamma \mf{x}_1\rangle \\[0.3em]
&- \tilde{k}_0 \|\nabla U(\mf{x}_0) - \mf{x}_1\|^2 \\[0.3em]
&+ \langle (1+\beta)\nabla U^*(\mf{x}_1) - \mf{x}_0,\; -\tilde{k}_1\big(\mathcal{S}(\mf{x}_0) + \mf{d}/k_1 \big)\rangle.
\end{aligned}
\]

First, consider the $\gamma$-terms:
\[
\begin{aligned}
&\langle \nabla U(\mf{x}_0) - \mf{x}_1, -\gamma \mf{x}_0\rangle 
+ \langle (1+\beta)\nabla U^*(\mf{x}_1) - \mf{x}_0, -\gamma \mf{x}_1\rangle  \\
&= -\gamma \Big( \langle \nabla U(\mf{x}_0), \mf{x}_0\rangle 
   + (1+\beta)\langle \nabla U^*(\mf{x}_1), \mf{x}_1\rangle 
   - 2\langle \mf{x}_0, \mf{x}_1\rangle \Big).
\end{aligned}
\]
By Lemma~\ref{le:properties}, items~\ref{prop:euler} and~\ref{prop:gamma} in the Appendix,
\[
\langle \nabla U(\mf{x}_0), \mf{x}_0\rangle \;\ge\; U(\mf{x}_0), 
\qquad 
\langle \nabla U^*(\mf{x}_1), \mf{x}_1\rangle \;\ge\; U^*(\mf{x}_1),
\]
so
\[
\begin{aligned}
&\langle \nabla U(\mf{x}_0) - \mf{x}_1, -\gamma \mf{x}_0\rangle 
+ \langle (1+\beta)\nabla U^*(\mf{x}_1) - \mf{x}_0, -\gamma \mf{x}_1\rangle \\
&\le -\gamma\big(U(\mf{x}_0) + (1+\beta)U^*(\mf{x}_1) - 2\langle \mf{x}_0,\mf{x}_1\rangle\big) \;\le\; 0.
\end{aligned}
\]

Hence, the derivative of $V$ satisfies
\begin{equation}
\label{eq:vdot}
\dot{V} \;\le\; -\tilde{k}_0\Gamma(\mf{x}) + \Pi(\mf{x}),
\end{equation}
where $\mf{x} = [\mf{x}_0,\mf{x}_1]^\top$ and
\begin{equation}
\label{eq:functions}
\begin{aligned}
\Gamma(\mf{x}) &= \|\nabla U(\mf{x}_0) - \mf{x}_1\|^2, \\[0.3em]
\Pi(\mf{x}) &= \sup_{\mf{d}\in\mathcal{D}}
\Big\langle (1+\beta)\nabla U^*(\mf{x}_1) - \mf{x}_0,\;
-\tilde{k}_1\Big(\mathcal{S}(\mf{x}_0)+\frac{1}{k_1}\mf{d}\Big)\Big\rangle.
\end{aligned}
\end{equation}

Next, we verify that $\Pi(\mf{x})<0$ whenever $\Gamma(\mf{x})=0$, i.e., for 
\[
\mathcal{X}_U = \{(\mf{x}_0,\mf{x}_1)\in\mathcal{X}\times\mathcal{X} : \nabla U(\mf{x}_0) = \mf{x}_1\}.
\]
{For any $\mf{x}\in\mathcal{X}_U$, we have $\mf{x}_0 = \nabla U^*(\mf{x}_1)$ since $(\nabla U)^{-1}=\nabla U^*$ \citep[Theorem 26.5]{rockafellar1970}. }
Hence
\[
(1+\beta)\nabla U^*(\mf{x}_1) - \mf{x}_0 = \beta \mf{x}_0, \qquad \mf{x}\in\mathcal{X}_U.
\]
Therefore,
\[
\begin{aligned}
\Pi(\mf{x})\big|_{\mf{x}\in\mathcal{X}_U}
&= -\tilde{k}_1 \beta \langle \mf{x}_0, \mathcal{S}(\mf{x}_0)\rangle 
   + \frac{\beta\tilde{k}_1}{k_1}\sup_{\mf{d}\in\mathcal{D}}\langle \mf{x}_0,\mf{d}\rangle \\[0.4em]
&\le -\tilde{k}_1 \beta c_\mathcal{S}\|\mf{x}_0\|
   + \frac{\beta\tilde{k}_1}{k_1}\|\mf{x}_0\| \sup_{\mf{d}\in\mathcal{D}}\|\mf{d}\| \\[0.4em]
&= -\beta \tilde{k}_1 \|\mf{x}_0\|\Big(c_\mathcal{S} - \frac{1}{k_1}\Big) \;<\; 0, \qquad \forall\mf{x}_0\ne \mf{0},
\end{aligned}
\]
where we have used the properties iii) and iv) in Assumption~\ref{as:basic}. The strict inequality follows from the condition $k_1>1/c_\mathcal{S}$ in \eqref{eq:gains}. 

Moreover, since $U$ and $U^*$ are homogeneous of degrees $3/2$ and $3$ by Assumption~\ref{as:basic} and Lemma~\ref{le:properties}–\ref{prop:hom} in the Appendix, their gradients are homogeneous of degrees $1/2$ and $2$, respectively.  
It follows that $\Gamma$ and $\Pi$ are $\mf{r}$-homogeneous of degree $2$ with $\mf{r} = [2\mathds{1}^\top, \mathds{1}^\top]$.  
Hence, by Lemma~\ref{le:gain:ineq} in the Appendix, there exists $c>0$ such that
\begin{equation}
\label{eq:c:const}
-\tilde{k}_0\Gamma(\mf{x}) + \Pi(\mf{x}) \;\le\; -c\|\mf{x}\|_{\mf{r}}^2,
\end{equation}
Therefore,
\[
\dot{V} \;\le\; -c\|\mf{x}\|_{\mf{r}}^2,
\]
which already guarantees asymptotic stability for all initial conditions $(\mf{x}_0(0),\mf{x}_1(0))\in\mathcal{X}\times\mathcal{X}$.  

Let $V_2(\mf{x})=\|\mf{x}\|_{\mf{r}}^2$ and $V_1(\mf{x})=V(\mf{x})$.  
By Lemma~\ref{le:comparison} in the Appendix,
\begin{equation}
\label{eq:v:const}
\|\mf{x}\|_{\mf{r}}^2 \;\ge\; \underline{v} V(\mf{x})^{2/3},
\end{equation}
for some $\underline{v}>0$.  
Hence,
\begin{equation}
\label{eq:lyap:ineq}
\dot{V} \;\le\; -c\underline{v}V^{2/3}.
\end{equation}

Applying the comparison lemma yields
\[
-\int_{V(\mf{x}(0))}^{0} V^{-2/3}\mathrm{d}V 
\;\ge\; c\underline{v}\int_0^T \mathrm{d}t,
\]
where $T$ denotes the settling time for $V$ to reach zero.  
Equivalently,
\[
T \;\le\; \frac{1}{3c\underline{v}} V(\mf{x}(0))^{1/3}.
\]
Thus, for any initial conditions $\mf{x}(0)\in\mathcal{X}\times \mathcal{X}$, all trajectories converge to the origin before a finite time depending on $V(\mf{x}(0))$. This completes the proof of Theorem \ref{th:abstract:st}.

\section{Proof of the Main Result}
\label{sec:proof}
{To prove the main result, we will relate the error system in \eqref{eq:error:system} with the abstract super-twisting in \eqref{eq:abstract:sys}.} Let
\begin{equation}
\label{eq:X:D}
\begin{aligned}
\mathcal{X} &= \{\mf{x}\in\mathbb{R}^{N} : \mf{x} = \mf{P}\mf{x}' \ \text{for some } \mf{x}'\in\mathbb{R}^{N}\}, \\
\mathcal{D} &= \{\mf{d}\in\mathcal{X} : \|\mf{d}\|\leq 1\}.
\end{aligned}
\end{equation}
Under Assumption~\ref{as:bound}, the disturbance $\mf{d}(t)$ from \eqref{eq:error:system} satisfies $\mf{d}(t)\in L\mathcal{D}$, thereby satisfying the conditions on $\mathcal{X},\mathcal{D}$ in Assumption~\ref{as:basic}. The error dynamics \eqref{eq:error:system} can thus be expressed as the differential inclusion
\begin{equation}
\label{eq:redcho:inclusion}
\begin{aligned}
\dot{\mf{e}}_{0}(t) &= -k_{0}\sqrt{L}\,\mf{D}\,\sgn{\mf{D}^\top \mf{e}_{0}(t)}{\frac{1}{2}}
    + \mf{e}_{1}(t) - \gamma \mf{e}_{0}(t), \\[1em]
\dot{\mf{e}}_{1}(t) &\in -k_{1}L\,\mf{D}\,\sgn{\mf{D}^\top \mf{e}_{0}(t)}{0}
    - \gamma \mf{e}_{1}(t) + L\mathcal{D}.
\end{aligned}
\end{equation}

Define $\mathcal{S}(\mf{e}_0) := \mf{D}\,\sgn{\mf{D}^\top \mf{e}_{0}}{0}$. Then,
$$ 
\begin{aligned} 
&\langle \mf{e}_0, \mathcal{S}(\mf{e}_0)\rangle = \langle \mf{e}_0, \mf{D}\sgn{\mf{D}^\top \mf{e}_{0}(t)}{0}\rangle \\ 
&=\langle \mf{D}^\top\mf{e}_0, \sgn{\mf{D}^\top \mf{e}_{0}}{0} \rangle = \langle \mf{v}, \sgn{\mf{v}}{0}\rangle = \sum_{\ell=1}^{|\mathcal{E}|} v_\ell \sgn{v_\ell}{0} \\ 
&= \sum_{\ell=1}^{|\mathcal{E}|} |v_\ell| \geq \|\mf{v}\|\geq \|\mf{D}^\top\mf{e}_0\|=\sqrt{\mf{e}_0^\top \mf{D}\mf{D}^\top\mf{e}_0}\\ 
&\geq \sqrt{\lambda_\mathcal{G}}\|\mf{e}_0\|
\end{aligned} 
$$
with $\mf{v} = \mf{D}^\top \mf{e}_0$ and $c_\mathcal{S} = \sqrt{\lambda_\mathcal{G}}$, where $\lambda_{\mathcal{G}}$ denotes the algebraic connectivity of $\mathcal{G}$, i.e., the smallest nonzero eigenvalue of the Laplacian $\mf{L}=\mf{D}\mf{D}^\top$. Thus, $\mathcal{S}$ satisfies the requirements in Assumption~\ref{as:basic}.

Now separate $\mf{D} = [\mf{d}_1,\dots,\mf{d}_{|\mathcal{E}|}]$, with each column $\mf{d}_\ell$ corresponding to an edge, and define
$$
U(\mf{e}_0) := \frac{2}{3}\sum_{\ell=1}^{|\mathcal{E}|} \big|\mf{d}_\ell^\top \mf{e}_0\big|^{\frac{3}{2}}.
$$
The function $U$ is positive, and $U(\mf{e}_0)=0$ only if $\mf{e}_0\in\ker(\mf{D}^\top)$. Since $\mf{e}_0\in\mathcal{X}$, this implies $\mf{e}_0=0$, so $U$ is positive definite on $\mathcal{X}$. Moreover, $|\bullet|^{\frac{3}{2}}$ is strictly convex, and so are non-singular affine combinations over $\mathcal{X}$ and compositions thereof. Hence, $U$ is strictly convex.

The gradient of $U$ is
$$
\begin{aligned}
&\nabla U(\mf{e}_0) 
= \sum_{\ell=1}^{|\mathcal{E}|} \mf{d}_\ell \,|\mf{d}_\ell^\top \mf{e}_0|^{\frac{1}{2}}
     \,\text{sign}(\mf{d}_\ell^\top\mf{e}_0) \\[0.5em]
&= \begin{bmatrix} \mf{d}_1,\dots,\mf{d}_{|\mathcal{E}|}\end{bmatrix}
   \begin{bmatrix}
   \sgn{\mf{d}_1^\top\mf{e}_0}{\frac{1}{2}} \\[-0.3em]
   \vdots \\[-0.3em]
   \sgn{\mf{d}_{|\mathcal{E}|}^\top\mf{e}_0}{\frac{1}{2}}
   \end{bmatrix} = \mf{D}\,\sgn{\mf{D}^\top\mf{e}_0}{\frac{1}{2}}.
\end{aligned}
$$
Thus, $U$ satisfies all conditions in Assumption~\ref{as:basic}, and \eqref{eq:redcho:inclusion} matches the form of the abstract super-twisting system \eqref{eq:abstract:sys}. By Theorem~\ref{th:abstract:st}, the consensus errors converge to the origin in finite time for gains selected as in \eqref{eq:gains}. Convergence of the consensus errors implies asymptotic convergence of the consensus components to $\bar{s}(t)$ and $\dot{\bar{s}}(t)$, as shown in \cite{redcho}, which completes the proof.

\begin{rem}
\label{rem:gains}
The gain conditions in \eqref{eq:gains} can be evaluated systematically by casting them as constrained nonlinear optimization problems solvable with standard tools. Specifically, \eqref{eq:dual} is used to compute the dual function $U^*$ over $\mathcal{X}$, while the lower bound of $\tilde{k}_1$ in \eqref{eq:gains} is obtained over the homogeneous ball $\|\mf{x}\|_{\mf{r}}=1$. Since feasibility requires gains only to exceed these bounds, one can compute them for a family of graphs and then select global gains as their maximum, ensuring validity across the entire family.
\end{rem}

\section{Event-Triggered Protocols}
\label{sec:et}

Since the distributed differentiator is intended to be based on communicating information between agents, it makes sense to study implementations of it that make efficient use of the communication network. Event-triggered algorithms are based on the principle of replacing continuous information exchange by a sequence of transmission instants, determined by local triggering rules. {Between two successive events, each agent evolves its local dynamics using the most recently received information from its neighbors, and instantaneous information of its local state. }

{In this section, we aim to derive an event-triggered version of the distributed differentiator and study possible tradeoffs between accuracy and communication effort. For the sake of generality, we will consider  that each edge $(i,j)\in\mathcal{E}$ is equipped with a sequence of transmission instants $\{\tau^{ij}_k\}_{k=0}^{\infty}$, at which both agents $i$ and $j$ broadcast their current outputs. The dynamics at agent $i$ then become
\begin{equation}
\label{eq:trigger:system}
\begin{aligned}
\dot{{\upeta}}_{i,0}(t) &= k_{0}\sqrt{L} \sum_{j\in\mathcal{N}_i}\sgn{\hat{s}_{i,0}(\tau^{ij}_k) - \hat{s}_{j,0}(\tau^{ij}_k)}{1/2} \\ &
+ {\upeta}_{i,1}(t) - \gamma {\upeta}_{i,0}(t),\\[0.8em]
\dot{{\upeta}}_{i,1}(t) &= k_1 L \sum_{j\in\mathcal{N}_i}\sgn{\hat{s}_{i,0}(\tau^{ij}_k) - \hat{s}_{j,0}(\tau^{ij}_k)}{0} 
- \gamma {\upeta}_{i,1}(t),
\end{aligned}
\end{equation}
where $\tau^{ij}_k$ denotes the most recent transmission time associated with edge $(i,j)$.}

{Therefore, the time sequence design depends on a recursive trigger rule of the form:
\begin{equation}
\label{eq:trigger:general}
\begin{aligned}
&\tau^{ij}_{k+1} =\inf\Big\{\, t\geq \tau^{ij}_{k} : \\&
\max\big\{|\hat{s}_{i,0}(t)-\hat{s}_{i,0}(\tau^{ij}_{k})|,\;
|\hat{s}_{j,0}(t)-\hat{s}_{j,0}(\tau^{ij}_{k})|\big\}\;\geq\;\delta_{ij}(t) \,\Big\},
\end{aligned}
\end{equation}
with $\tau^{ij}_0=0$. Each agent monitors only its own output and broadcasts when the deviation exceeds $\delta_{ij}(t)$, which simultaneously updates both ends resulting in $\tau^{ij}_k=\tau^{ji}_k$.}

{We write $\delta_{ij}(t)$ as time-varying to cover a constant send-on-delta trigger, state dependent triggers and time varying dynamic triggers. In the rest of this section, we discuss the consequences of choosing $\delta_{ij}(t)$ in several of these forms.}

{First, we start with a design that ensures the exclusion of Zeno behavior, consisting of the accumulation of events in finite-time, and the exact asymptotic convergence of the distributed differentiator.}

\begin{prop}
\label{prop:trigger:asymp}
{Let \eqref{eq:trigger:general} with $\delta_{ij}(t)>0, \forall t\geq 0$,  $\lim_{t\to\infty} \delta_{ij}(t)=0$, $\delta_{ij}(t)$ bounded and attaining a strictly positive minimum for every compact time interval. 
Then, there is no Zeno behavior for the sequence $\{\tau^{ij}_k\}_{k=0}^\infty$ and all pairs of agents $(i,j)\in\mathcal{E}$ obtain
$$
\lim_{t\to\infty}|\hat{s}_{i,0}(t)-\hat{s}_{j,0}(t)| =0,\quad
\lim_{t\to\infty}|\hat{s}_{i,1}(t)-\hat{s}_{j,1}(t)| =0.
$$}
\end{prop}
\begin{pf}
{Note that, for every compact set $[0,T]$ with $T>0$, we have $\delta_{ij}(t)>\underline{\delta},\, \forall t\in[0,T]$ for a suitable $\underline{\delta}$ depending on $T$. Moreover, \eqref{eq:trigger:system} cannot have escapes in finite time as it is comprised of a linear system and piecewise continuous bounded terms over the compact interval $[0,T]$. Thus, on such interval, all $|\dot{\hat{s}}_{i,0}(t)|\leq B, \forall t\in[0,T]$, for a suitable bound $B$ depending on $T$. Therefore,
$$
|\hat{s}_{i,0}(\tau_{k+1}^{ij})-\hat{s}_{i,0}(\tau_{k}^{ij})|\leq \left|\int_{\tau_{k}^{ij}}^{\tau_{k+1}^{ij}}\dot{\hat{s}}_{i,0}(\tau)\text{d}\tau\right|\leq B(\tau_{k+1}^{ij}-\tau_k^{ij}).
$$
On the other hand,
$$
|\hat{s}_{i,0}(\tau_{k+1}^{ij})-\hat{s}_{i,0}(\tau_{k}^{ij})|\geq \delta_{ij}(\tau_{k+1}^{ij})\geq \underline{\delta},
$$
implying $\tau_{k+1}^{ij}-\tau_k^{ij}\geq \underline{\delta}/B>0$. Consequently, a strictly positive minimum inter-event time is ensured on every compact interval, and Zeno behavior in the form of accumulation of events in finite time cannot occur. Finally, note that for all $t\geq T,\ T>0$, we have $\delta_{ij}(t)\leq \overline{\delta}$, where the bound $\overline{\delta}$ depends on $T$. Hence, we can write
$$
\hat{s}_{i,0}(\tau^{ij}_k) - \hat{s}_{j,0}(\tau^{ij}_k) = \hat{s}_{i,0}(t) - \hat{s}_{j,0}(t) + \varepsilon_{ij}(t)
$$
with $|\varepsilon_{ij}(t)|\leq 2\overline{\delta},\, \forall t\geq T$. As a result, \eqref{eq:trigger:system} is a disturbed version of \eqref{eq:redcho}. Therefore, standard homogeneity arguments as in \cite{aldana2023} are used resulting in practical convergence, with accuracy scaling as $\limsup_{t\to\infty}|\hat{s}_{i,0}(t)-\hat{s}_{j,0}(t)| \leq c_0 \overline{\delta}$ and $ \limsup_{t\to\infty}|\hat{s}_{i,1}(t)-\hat{s}_{j,1}(t)| \leq c_1\sqrt{\overline{\delta}}$ with suitable constants $c_0,c_1>0$ that do not depend on $\overline{\delta}$. Since $\overline{\delta}\to 0$ as $T\to \infty$, then the differentiator consensus errors converge to 0 as well.\qed}
\end{pf}

{One possible choice for $\delta_{ij}(t)$ to comply with the conditions of Proposition \ref{prop:trigger:asymp} is
$$
\delta_{ij}(t) = \delta\exp(-q_{ij}(t-p_{ij}))
$$
with $\delta,q_{ij}>0,p_{ij}\geq 0$ being constants decided by the pair of agents $(i,j)$.}

{While Proposition \ref{prop:trigger:asymp} establishes asymptotic exactness without Zeno behavior, the next proposition shows that this remains compatible with the fact that triggering thresholds $\delta_{ij}(t)$ vanishing asymptotically necessarily lead to inter-event times that also vanish asymptotically.
}

\begin{prop}
\label{prop:trigger:impossibility}
{Let \eqref{eq:trigger:general} with $\delta_{ij}(t)$ satisfying the conditions stated in Proposition \ref{prop:trigger:asymp}. Assume that there exists $c>0,\, T_0>0$ such that, for all $t\geq 0$ and all $T\in(0,T_0]$, we have:
\begin{equation}
\label{eq:avgspeed}
\left|\frac{1}{T}\int_{t}^{t+T}\dot{\hat{s}}_{i,0}(\tau)\text{d}\tau\right|\geq c.
\end{equation}
Then, for every \( T \in (0, T_0] \) there exists \( K \ge 0 \) such that, for all \( k \ge K \), it follows that
\[
\tau_{k+1}^{ij} - \tau_k^{ij} \le T .
\]
Equivalently, there exists no strictly positive minimum inter-event time for the time interval $[0,\infty)$.}
\end{prop}
\begin{pf}
{Define the set
    $$
    \begin{aligned}
    &\mathcal{T} := \Big\{\, t\geq\tau^{ij}_{k} : \\&
\max\big\{|\hat{s}_{i,0}(t)-\hat{s}_{i,0}(\tau^{ij}_{k})|,\;
|\hat{s}_{j,0}(t)-\hat{s}_{j,0}(\tau^{ij}_{k})|\big\}\;\geq\;\delta_{ij}(t) \,\Big\}.
\end{aligned}
    $$ Given $T_0>0$, choose any $T\in(0,T_0]$. Choose $T_{c}\geq 0$ such that $\delta_{ij}(t)\leq c T$ for all $t\geq T_c$. Choose $k\geq 0$ such that $\tau_{k}^{ij}\geq T_c$. Note that $\delta_{ij}(\tau_k^{ij}+T)\leq c T$. Hence, using \eqref{eq:avgspeed}, it follows that
    $$
    \left|\int_{t}^{t+T}\dot{\hat{s}}_{i,0}(\tau)\text{d}\tau\right|=|\hat{s}_{i,0}(t)-\hat{s}_{i,0}(\tau_k^{ij})|\geq cT\geq \delta_{ij}(t) 
    $$
    for $t=\tau_k^{ij}+T$.
Then, $\tau_{k}^{ij}+T\in\mathcal{T}$ and $\tau_{k+1}^{ij}=\inf \mathcal{T}$. Thus, $\tau_{k}^{ij}+T\geq \tau_{k+1}^{ij}$, equivalently $T\geq \tau_{k+1}^{ij} - \tau_k^{ij}$, completing the proof. \qed}
\end{pf}

{Note that the condition in \eqref{eq:avgspeed} implies the presence of an average persistent variation of the outputs \( \hat{s}_{i,0}(t) \), which may occur when the inputs \( s_i(t) \) are not asymptotically constant. It is precisely in this case that triggering must become increasingly frequent as \( \delta_{ij}(t) \) vanishes.
}

{Motivated by Proposition \ref{prop:trigger:impossibility}, and since we are interested in persistently varying references that may induce \eqref{eq:avgspeed}, we propose a different threshold design that ensures a strictly positive minimum inter-event time for all \( t \in [0,\infty) \) along each trajectory. The proposed trigger is state dependent, allowing communication to be reduced during the transient, at the cost of necessarily losing asymptotic exactness, while still providing a favorable tradeoff between performance and communication.}

\begin{thm}
\label{prop:trigger}
{Let the assumptions and conditions of Theorem~\ref{th:main} hold, and consider the event-triggered protocol in \eqref{eq:trigger:system} under the trigger rule 
\begin{equation}
\label{eq:trigger:state}
\begin{aligned}
&\tau^{ij}_{k+1} =\inf\Big\{\, t\geq \tau^{ij}_{k} : 
\max\big\{|\hat{s}_{i,0}(t)-\hat{s}_{i,0}(\tau^{ij}_{k})|,\\&
|\hat{s}_{j,0}(t)-\hat{s}_{j,0}(\tau^{ij}_{k})|\big\}\geq\;\delta +\sigma |\hat{s}_{i,0}(\tau^{ij}_{k})-\hat{s}_{j,0}(\tau^{ij}_{k})|\,\Big\},
\end{aligned}
\end{equation}
for some arbitrary $\delta>0$ and appropriate $\sigma\geq 0$. Then,
\begin{enumerate}
    \item For every initial condition, there exists a strictly positive minimum inter-event-time $\underline{\tau}>0$ such that $\tau_{k+1}^{ij}-\tau_{k}^{ij}\geq \underline{\tau}$, for all $(i,j)\in\mathcal{E}, k\geq 0$.
    \item There exist constants $c_0,c_1>0$, independent of $\delta$ such that
\begin{equation}
\begin{aligned}
\label{eq:accuracy}
&\limsup_{t\to\infty}|\hat{s}_{i,0}(t)-\hat{s}_{j,0}(t)| \;\leq\; c_0\,\delta,\\
&\limsup_{t\to\infty}|\hat{s}_{i,1}(t)-\hat{s}_{j,1}(t)| \;\leq\; c_1\,\sqrt{\delta},
\end{aligned}
\end{equation}
for all $(i,j)\in\mathcal{E}$.
\end{enumerate}}
\end{thm}
{The proof of Theorem \ref{prop:trigger} is given in detail in Section \ref{sec:trigger:proof}, taking advantage of the abstract super-twisting-based Lyapunov analysis provided in this work.}

\begin{rem}
\label{rem:event}
{The constants $c_0,c_1$ in Theorem~\ref{prop:trigger} depend only on the selected gains. Unlike \cite{aldana2023}, where their existence was only guaranteed, here they can be estimated using the Lyapunov function \eqref{eq:lyapunov} from Theorem~\ref{th:main}. Following the approach for homogeneous differentiators \citep{Cruz2019}, one computes the smallest $\theta\geq 0$ such that the level set $\mathcal{L}_\theta=\{\mf{x}\in\mathcal{X}:V(\mf{x})\leq \theta\}$ satisfies $\dot V<0$ on its boundary, with $\dot V$ taken from \eqref{eq:lyap:trigger:dot} after replacing $U(\mf{x})$, $U^*(\mf{x})$, and $\mathcal{S}(\mf{x})$ by their perturbed versions $U(\mf{x}+\bm{\varepsilon})$, $U^*(\mf{x}+\bm{\varepsilon})$, and $\mathcal{S}(\mf{x}+\bm{\varepsilon})$ under $\bm{\varepsilon}\in[0,2\delta]^{|\mathcal{E}|}$, due to the trigger effect with normalized $\delta=1, \sigma=0$. This ensures trajectories reach and remain in $\mathcal{L}_\varepsilon$. The constants $c_0,c_1$ are then obtained as
\[
c_0=\sup_{\mf{x}\in\mathcal{L}_\theta}\|\mf{x}_0\|
\qquad
c_1=\sup_{\mf{x}\in\mathcal{L}_\theta} \|\mf{x}_1\|
\]
While such numerical procedure requires extensive use of constrained nonlinear optimization programs, this is a much more principled design approach than what was previously found in the literature \citep{edcho, redcho, aldana2023}, where gains $k_0,k_1$ and constants $c_0,c_1$ were found approximately by means of simulations with no real guarantee of their effectiveness.}
\end{rem}

\subsection{{Analysis under event-triggered communication}}
\label{sec:trigger:proof}
{In this section, we provide the proof of Theorem \ref{prop:trigger}. The proof strategy is as follows. First, Lemma \ref{lem:trigger:change} rewrites the error dynamics of the event-triggered system as a disturbed version of the abstract super-twisting error system. Then, in Lemma \ref{le:lyap:trigger} we show that a suitable Lyapunov inequality is satisfied by the disturbed system, which implies finite time stability provided that solutions exist for all \( t \). Lemma \ref{le:miet} establishes the existence of a strictly positive minimum inter-event time for all \( t \), thereby ensuring existence of solutions. These arguments are combined at the end of the section to conclude finite time stability toward the accuracy bounds stated in Theorem \ref{prop:trigger}.}

\begin{lem}
\label{lem:trigger:change}
    {Consider $\mf{x}_0(t),\mf{x}_1(t)$ obtained through the change of variables \eqref{eq:variables:hat:s}, \eqref{eq:variables:e}, \eqref{eq:variables:x} applied to system \eqref{eq:trigger:system} under the trigger rule \eqref{eq:trigger:state} with $\sigma<1/2$. Thus, $\mf{x}_0(t),\mf{x}_1(t)$ have dynamics complying with:
    \begin{equation}
        \label{eq:et:x}
\begin{aligned}
\dot{\mf{x}}_0(t) &= -\gamma \mf{x}_0(t) - \tilde{k}_0\left(\mf{D}\,\sgn{\mf{D}^\top \mf{x}_{0}(t) + \bm{\varepsilon}(t)}{\frac{1}{2}} - \mf{x}_1(t)\right), \\[0.4em]
\dot{\mf{x}}_1(t) &\in -\gamma \mf{x}_1(t) - \tilde{k}_1\left(\mf{D}\,\sgn{\mf{D}^\top \mf{x}_{0}(t)+\bm{\varepsilon}(t)}{0} + \frac{1}{k_1}\mathcal{D}\right),
\end{aligned}
    \end{equation}
    where $\bm{\varepsilon}(t)\in\text{image}(\mf{D}^\top)$ and
    \begin{equation}
    \label{eq:bound:epsilon}
    \|\bm{\varepsilon}(t)\|\leq \frac{2}{1-2\sigma}\left(\frac{\delta\sqrt{|\mathcal{E}|}}{L}+\sigma\|\mf{D}^\top{\mf{x}}_0(t)\|\right),\, \quad \forall t\geq 0.
    \end{equation}}
\end{lem}
\begin{pf}
{Start by writing
$$
\hat{s}_{i,0}(\tau^{ij}_{k})-\hat{s}_{j,0}(\tau^{ij}_{k}) = \hat{s}_{i,0}(t)-\hat{s}_{j,0}(t) + \varepsilon_{i}(t)-\varepsilon_{j}(t),
$$
where $\varepsilon_{i}(t) = \hat{s}_{i,0}(\tau^{ij}_{k})-\hat{s}_{i,0}(t)$ and similarly for $\varepsilon_{j}(t)$. Moreover, due to \eqref{eq:trigger:state},
$$
|\varepsilon_{i}(t)|\leq \delta + \sigma |\hat{s}_{i,0}(\tau^{ij}_{k})-\hat{s}_{j,0}(\tau^{ij}_{k})|, \ \forall t \in [\tau_k^{ij}, \tau_{k+1}^{ij}).
$$ Therefore, we have
$$
\begin{aligned}
&|\hat{s}_{i,0}(\tau^{ij}_{k})-\hat{s}_{j,0}(\tau^{ij}_{k})| = |\hat{s}_{i,0}(t)-\hat{s}_{j,0}(t) + \varepsilon_{i}(t)-\varepsilon_{j}(t)| \\
&\leq |\hat{s}_{i,0}(t)-\hat{s}_{j,0}(t)| + |\varepsilon_{i}(t)|+|\varepsilon_{j}(t)|\\ &\leq |\hat{s}_{i,0}(t)-\hat{s}_{j,0}(t)| + 2\delta + 2\sigma |\hat{s}_{i,0}(\tau^{ij}_{k})-\hat{s}_{j,0}(\tau^{ij}_{k})|,
\end{aligned}
$$
which implies
$$
|\hat{s}_{i,0}(\tau^{ij}_{k})-\hat{s}_{j,0}(\tau^{ij}_{k})| \leq \frac{|\hat{s}_{i,0}(t)-\hat{s}_{j,0}(t)| + 2\delta}{1-2\sigma}.
$$
This result enables us to write
$$
\begin{aligned}
&\max\{|\varepsilon_i(t)|,|\varepsilon_j(t)|\} \leq  \delta + \sigma |\hat{s}_{i,0}(\tau^{ij}_{k})-\hat{s}_{j,0}(\tau^{ij}_{k})| \\
&\leq \delta + \sigma\left(\frac{|\hat{s}_{i,0}(t)-\hat{s}_{j,0}(t)|}{1-2\sigma} + \frac{2\delta}{1-2\sigma}\right)\\
&=\frac{\delta}{1-2\sigma} + \frac{\sigma}{1-2\sigma}|\hat{s}_{i,0}(t)-\hat{s}_{j,0}(t)|.
\end{aligned}
$$
Thus, we can obtain
$$
\begin{aligned}
&\hat{s}_{i,0}(\tau^{ij}_{k})-\hat{s}_{j,0}(\tau^{ij}_{k}) = \hat{s}_{i,0}(t)-\hat{s}_{j,0}(t) + \varepsilon_{ij}(t),
\end{aligned}
$$
with $\varepsilon_{ij}(t) = \varepsilon_i(t)-\varepsilon_j(t)$ and
\begin{equation}
\label{eq:epsilon:bound}
|\varepsilon_{ij}(t)|\leq \frac{2\delta}{1-2\sigma} + \frac{2\sigma}{1-2\sigma}|\hat{s}_{i,0}(t)-\hat{s}_{j,0}(t)|.
\end{equation}
Using the equalities above, we write \eqref{eq:trigger:system} as
\begin{equation}
\begin{aligned}
\dot{{\upeta}}_{i,0}(t) &= k_{0}\sqrt{L} \sum_{j\in\mathcal{N}_i}\sgn{\hat{s}_{i,0}(t)-\hat{s}_{j,0}(t) + \varepsilon_{ij}(t)}{1/2} \\ &
+ {\upeta}_{i,1}(t) - \gamma {\upeta}_{i,0}(t),\\[0.8em]
\dot{{\upeta}}_{i,1}(t) &= k_1 L \sum_{j\in\mathcal{N}_i}\sgn{\hat{s}_{i,0}(t)-\hat{s}_{j,0}(t) + \varepsilon_{ij}(t)}{0} 
- \gamma {\upeta}_{i,1}(t),
\end{aligned}
\end{equation}
which is equivalent to the protocol dynamics in \eqref{eq:redcho} with the added disturbance terms $\varepsilon_{ij}(t)$. Taking into account that $\mf{x}_0(t)=\mf{e}_0(t)/L$, we define a disturbance vector $\bm{\varepsilon}(t)=\mf{D}^\top[\varepsilon_1(t),\dots,\varepsilon_N(t)]^\top/L$, which fulfills
$$
\begin{aligned}
    &\|\bm{\varepsilon}(t)\| = \sqrt{\sum_{(i,j)\in\mathcal{E}}(\varepsilon_{ij}(t)/L)^2} \\ &\leq\frac{2}{(1-2\sigma)L}\sqrt{\sum_{(i,j)\in\mathcal{E}}(\delta + \sigma|\hat{s}_{i,0}(t)-\hat{s}_{j,0}(t)| )^2} \\
    &\leq \frac{2}{(1-2\sigma)L}\left({\delta\sqrt{|\mathcal{E}|}} + \sigma\sqrt{\sum_{(i,j)\in\mathcal{E}}(\hat{s}_{i,0}(t)-\hat{s}_{j,0}(t) )^2}\right)\\
    &=\frac{2}{(1-2\sigma)L}\left({\delta\sqrt{|\mathcal{E}|}}+\sigma\|\mf{D}^\top\hat{\mf{s}}_0(t)\|\right)\\
    &=\frac{2}{1-2\sigma}\left(\frac{\delta\sqrt{|\mathcal{E}|}}{L}+\sigma\|\mf{D}^\top{\mf{x}}_0(t)\|\right).
\end{aligned}
$$
Finally, the same process as in Section \ref{sec:proof} and Lemma \ref{lem:change}, lead to the form in \eqref{eq:et:x}, completing the proof. \qed}
\end{pf}
{We now provide the consequences in terms of Lyapunov analysis for the disturbed system.}
\begin{lem}
\label{le:lyap:trigger}
    {Let $\bm{\varepsilon}:\mathbb{R}_{\geq 0}\to\text{image}(\mf{D}^{\top})$ comply with \eqref{eq:bound:epsilon} under $\delta=0$. Moreover, set
    \begin{equation}
    \label{eq:sigma:condition}
        \sigma\leq \min\left\{\frac{1}{4},\frac{1}{2}\left(1+\left(\frac{c\underline{v}}{c_{\psi}}\right)^2\right)^{-1}\right\} 
    \end{equation}
    with $c,\underline{v},c_\psi>0$ obtained from \eqref{eq:c:const}, \eqref{eq:v:const} and \eqref{eq:psi:const} respectively. Thus, there exists $c_\sigma>0$ such that
    \begin{equation}
    \label{eq:lyap:trigger}
    \dot{V} \;\le\; -c_{\sigma} V^{2/3}.
    \end{equation}
    for $V$ defined in \eqref{eq:lyapunov} along solutions of \eqref{eq:et:x}.}
\end{lem}
\begin{pf}
    {First, write \eqref{eq:et:x} as
    \begin{equation}
\label{eq:abstract:disturbed}
\begin{aligned}
\dot{\mf{x}}_0 &= -\gamma \mf{x}_0 - \tilde{k}_0\big(\nabla U(\mf{x}_0) - \mf{x}_1\big)+\bm{\psi}_0(\mf{x}_0,\bm{\varepsilon}), \\[0.4em]
\dot{\mf{x}}_1 &\in -\gamma \mf{x}_1 - \tilde{k}_1\left(\mathcal{S}(\mf{x}_0) + \frac{1}{k_1}\mathcal{D}\right) + \bm{\psi}_1(\mf{x}_0,\bm{\varepsilon}),
\end{aligned}
\end{equation}
with
$$
\begin{aligned}
\bm{\psi}_0(\mf{x}_0,\bm{\varepsilon}) &:= \tilde{k}_0\mf{D}\left(\sgn{\mf{D}^\top \mf{x}_{0}}{\frac{1}{2}} - \sgn{\mf{D}^\top \mf{x}_{0} + \bm{\varepsilon}}{\frac{1}{2}}\right) \\
\bm{\psi}_1(\mf{x}_0,\bm{\varepsilon}) &:= \tilde{k}_1\mf{D}\left(\sgn{\mf{D}^\top \mf{x}_{0}}{0} - \sgn{\mf{D}^\top \mf{x}_{0} + \bm{\varepsilon}}{0}\right).
\end{aligned}
$$
Now, let $\mf{w}=\mf{D}^\top\mf{x}_0$ and denote with ${w}_{ij}$ the component of $\mf{w}$ corresponding to edge $(i,j)$. Assume $w_{ij}\geq 0$ to conclude
$$
{w}_{ij} + {\varepsilon}_{ij} \geq w_{ij} - \frac{2\sigma}{1-2\sigma}|w_{ij}| = \left(\frac{1-4\sigma}{1-2\sigma}\right)w_{ij}\geq 0.
$$
where we used \eqref{eq:epsilon:bound}. Likewise, assume $w_{ij}\leq 0$ to conclude $w_{ij} + \varepsilon_{ij}\leq0$. As a result:
$$
\text{sign}(w_{ij} + {\varepsilon}_{ij}) = \text{sign}(w_{ij}).
$$
Therefore, $\bm{\psi}_1(\mf{x}_0,\bm{\varepsilon})\equiv \mf{0}$ for all $\mf{x}_0$ and admissible $\bm{\varepsilon}$. For $\bm{\psi}_0(\mf{x}_0,\bm{\varepsilon})$ note that
$$
\begin{aligned}
&\left|\sgn{w_{ij}}{1/2}-\sgn{w_{ij}+\varepsilon_{ij}}{1/2}\right| = \left||w_{ij}|^{1/2} - |w_{ij}+\varepsilon_{ij}|^{1/2}\right|\\ &\leq \left|(w_{ij}) - (w_{ij}+\varepsilon_{ij})\right|^{1/2} = |\varepsilon_{ij}|^{1/2}
\end{aligned}
$$
since ${|\bullet|^{1/2}}$ is Hölder continuous. Therefore,
$$
\begin{aligned}
&\|\bm{\psi}_0(\mf{x}_0,\varepsilon)\| \leq \tilde{k}_0\|\mf{D}\| \left\| \sgn{\mf{D}^\top \mf{x}_{0}}{\frac{1}{2}} - \sgn{\mf{D}^\top \mf{x}_{0} + \bm{\varepsilon}}{\frac{1}{2}} \right\| \\
&\leq \tilde{k}_0\|\mf{D}\| \sqrt{\sum_{(i,j)\in\mathcal{E}}|\varepsilon_{ij}| } \leq \tilde{k}_0|\mathcal{E}|^{1/4}\|\mf{D}\|  \sqrt{\frac{2\sigma}{1-2\sigma}} \sqrt{\|\mf{D}^\top\mf{x}_0\|} \\
\end{aligned}
$$
As a result:
\begin{equation}
\label{eq:psi:const}
\begin{aligned}
    &\langle\nabla U(\mf{x}_0) - \mf{x}_1, \bm{\psi}_0(\mf{x}_0,\varepsilon)\rangle \leq \|\nabla U(\mf{x}_0) - \mf{x}_1\|\|\bm{\psi}_0(\mf{x}_0,\varepsilon)\| \\
    & \leq \tilde{k}_0|\mathcal{E}|^{1/4}\|\mf{D}\|  \sqrt{\frac{2\sigma}{1-2\sigma}} \|\nabla U(\mf{x}_0) - \mf{x}_1\|\sqrt{\|\mf{D}^\top\mf{x}_0\|} \\ 
    &\leq c_\psi\sqrt{\frac{2\sigma}{1-2\sigma}} V(\mf{x}_0,\mf{x}_1)^{2/3}
\end{aligned}
\end{equation}
for an appropriate $c_\psi\geq 0$, independent of $\sigma$, computed according to Lemma \ref{le:comparison} in the Appendix. Now, we compute the dynamics of $V$ following the procedure in Section \ref{sec:lyap:convergence}, in particular up to \eqref{eq:lyap:ineq}, and we obtain:
\begin{equation}
\label{eq:lyap:trigger:dot}
\begin{aligned}
&\dot{V} \leq -c\underline{v}V^{2/3} + \langle\nabla U(\mf{x}_0) - \mf{x}_1, \bm{\psi}(\mf{x}_0,\bm{\varepsilon})\rangle \\
&\leq -c\underline{v}V^{2/3} + c_\psi \sqrt{\frac{2\sigma}{1-2\sigma}} V^{2/3} \leq -c_\sigma V^{2/3} 
\end{aligned}
\end{equation}
with 
$$
c_\sigma := c\underline{v} - c_\psi\sqrt{\frac{2\sigma}{1-2\sigma}} >0
$$
by the choice of $\sigma$ in the statement of the lemma, completing the proof. \qed}
\end{pf}
{In the following result, we study the inter-event time of the protocol.}
\begin{lem}
\label{le:miet}
{Consider the conditions of Theorem \ref{prop:trigger}. Thus, for every initial condition, there exists a strictly positive minimum inter-event-time $\underline{\tau}>0$ such that $\tau_{k+1}^{ij}-\tau_{k}^{ij}\geq \underline{\tau}$, for all $(i,j)\in\mathcal{E}, k\geq 0$.}
\end{lem}
\begin{pf}
    {First, note that there are no escapes in finite time for \eqref{eq:trigger:system}. Therefore, $|\dot{\hat{s}}_{i,0}(t)|\leq B_s$ for all $t$ in a compact time interval $[0,T_0]$, for some $B_s$ which depends on the interval length $T_0$ and the initial conditions for $\hat{\mf{s}}_0(0) = [\hat{{s}}_{1,0}(0),\dots,\hat{{s}}_{N,0}(0)]^{\top}$. Assume, for a contradiction, that there exists $T>0$ such that $\lim_{k\to\infty}\tau^{ij}_k=T$. Thus, there exists $k\geq 0$ such that
    \begin{equation}
    \label{eq:no:miet}
    \tau_{k+1}^{ij} - \tau_k^{ij}\leq \frac{\delta}{2B_s}
    \end{equation}
    Therefore, for such $k$:
    $$
    \begin{aligned}
    &\delta\leq  |\hat{s}_{i,0}(\tau_{k+1}^{ij}) - \hat{s}_{i,0}(\tau_k^{ij})|  =    \left|\int_{\tau_k^{ij}}^{\tau_{k+1}^{ij}}\dot{\hat{s}}_{i,0}(\tau)\text{d}\tau\right| \\&  \leq (\tau_{k+1}^{ij} - \tau_k^{ij})B_s \leq \frac{\delta}{2}
    \end{aligned}
    $$
    which is a contradiction. Thus, $\lim_{k\to\infty}\tau_k^{ij}=\infty$. As a result, solutions to \eqref{eq:trigger:system} are maximally defined over $t\in[0,\infty)$ and no Zeno phenomenon occurs. Over the unbounded time interval $[0,\infty)$, Lemma \ref{le:lyap:trigger} implies $\dot{V}\leq 0$ from \eqref{eq:lyap:trigger}. Since, $V$ is positive definite and radially unbounded from Lemma \ref{le:gamma}, then it follows that $\mf{x}_0(t),\mf{x}_1(t)$ are bounded for all $t\geq 0$. Therefore, the right hand side of \eqref{eq:trigger:system} is upper bounded for all $t\geq 0$ and, as a result, $|\dot{\hat{s}}_{i,0}(t)|\leq B_s$ for some $B_s$ which depends only on the initial conditions. Finally, assume for a contradiction that there is no strictly positive minimum inter-event time. Therefore, there exists $k\geq 0$ such that \eqref{eq:no:miet} is complied, leading again to a contradiction, completing the proof. \qed}
\end{pf}

\begin{pf}[Of Theorem \ref{prop:trigger}]
{First, Lemma \ref{le:miet} shows that Proposition \ref{prop:trigger}-(1) is true. For item (2), note that $\delta=0$ can induce Zeno behavior only by allowing not maximally defined disturbances $\bm{\varepsilon}(t)$ in \eqref{eq:et:x}, in which these are not defined after some finite time $T$. For any other admissible disturbance defined over the unbounded time interval $t\in[0,\infty)$, Lemma \ref{le:lyap:trigger} implies \eqref{eq:lyap:trigger}, which in turn implies finite time convergence towards the origin when $\delta=0$. In this case, note that $\mf{D}^\top\mf{x}_0$ and the admissible set for $\bm{\varepsilon}$ in \eqref{eq:bound:epsilon} are both set valued functions of $\mf{x}_0$ with the same homogeneity degree. Therefore, allowing $\delta>0$ implies an additive bounded disturbance over $\bm{\varepsilon}$, which results in practical convergence, with accuracy scaling as \eqref{eq:accuracy} by following standard homogeneity arguments as in \cite{aldana2023} for the homogeneous differential inequality \eqref{eq:lyap:trigger}. \qed }
\end{pf}

\section{Discussion}
\label{sec:discussion}
{The results of this work highlight the structural and practical benefits of a Lyapunov analysis based on homogeneity for distributed differentiation. Homogeneity is central to the proposed approach, as it directly shapes the robustness properties of the distributed algorithm. Homogeneous sliding mode based differentiators achieve steady state accuracy bounds that scale with the square root of the disturbance magnitude. As shown in \citep{seeber2023a}, no differentiator can achieve a better asymptotic accuracy when estimating derivatives from measurements corrupted by disturbances of magnitude $\delta$.}

{Consistently with this fundamental limit, under bounded noise or triggering induced disturbances of magnitude $\delta$, the estimation error converges to a neighborhood of order $\sqrt{\delta}$. Note that in our case, such disturbances can come in the form of measurement noise, delays or event-triggered communication. The accuracy analysis was formally established for distributed differentiators of the REDCHO type under measurement noise in \citep{aldana2023}, with analysis under delays in \cite{aldana2025}, still without explicit computation of scaling constants. The present work shows that the same optimal scaling is preserved under properly designed event-triggered communication as in Theorem \ref{prop:trigger}.}

{It is instructive to compare this approach with a possible alternative using methods from the literature. One first computes the average of the input signals $s_i(t)$ using a DAC algorithm, and then differentiates the resulting average using a local differentiator. Average consensus can be implemented using linear protocols \citep{Solmaz2017} or exact first order sliding mode consensus schemes \citep{freeman2019}. When combined with event-triggered communication or measurement noise, these cascaded designs lose homogeneity. Linear consensus dynamics and nonlinear differentiation lead to mixed degree behavior, so homogeneity based robustness arguments no longer apply. In the linear case, additional steady state errors are introduced and combined with noise effects. In the first order sliding mode case, exact convergence can be achieved ideally, but chattering appears under noise and is further amplified when injected into the differentiator. As a result, the $\sqrt{\delta}$ accuracy scaling for derivative estimation under noise or triggering cannot be guaranteed. This structural loss directly limits the achievable performance guarantees of these competing approaches.}

{As for the limitations of the Lyapunov approach, note that the Lyapunov function applied to the distributed differentiator depends explicitly on the graph structure through the incidence matrix. As a result, for time-varying topologies the Lyapunov function itself would change over time, preventing the use of a single potential function $U$ to establish convergence. Moreover, for directed graphs, the distributed dynamics may not admit an underlying convex potential function, which is a key ingredient in the homogeneity-based Lyapunov analysis. For these reasons, extending the results to directed or time-varying graphs is not trivial and may require different analytical tools, which are left for future work.}

{We now discuss the event-triggered aspect of the proposed approach. A closely related work is \citep{xu2024}, which achieves asymptotic exactness for dynamic average consensus under event-triggered communication. That work relies on a boundary layer approximation of the sign function, which asymptotically recovers exact tracking and can be interpreted as employing a vanishing triggering threshold. In this sense, its result is closely aligned with Proposition \ref{prop:trigger:asymp}. Moreover, in \citep{xu2024}, Zeno behavior is excluded. However, the corresponding minimum inter-event time guarantees hold only over compact time intervals. This limitation is not incidental. It is a direct consequence of Proposition \ref{prop:trigger:impossibility}, which shows that when signals keep varying persistently and the triggering threshold vanishes asymptotically, a strictly positive minimum inter-event time over $[0,\infty)$ cannot exist. This fundamental limitation motivates the state dependent triggering strategy proposed in Theorem \ref{prop:trigger}, which enforces a strictly positive minimum inter-event time by trading asymptotic exactness for a controlled terminal accuracy.}

\section{Numerical Example}
\label{sec:sim}
We now illustrate the performance of the event-triggered distributed differentiator in \eqref{eq:trigger:system}.
The simulation was carried out using a forward Euler scheme with step $\Delta t=10^{-4}\mathsf{s}$ over a horizon of $T=10\mathsf{s}$. The network consists of $N=5$ agents arranged in a ring topology.  
Each agent measures a sinusoidal signal $s_i(t)=\sin(\omega_i t+\phi_i)$ with frequencies $\omega_i\in\{1.73,0.58,1.12,0.37,1.95\}$ and phases $\phi_i\in\{0.27,1.66,0.09,1.92,0.45\}$.  
The collective objective is to track the average derivative $\dot{\bar s}(t)$.  
The parameters are set to $k_0=4.0 > 3.45$, $k_1=13> 12.4$ complying with the conditions \eqref{eq:gains} where the lower bounds were found as in Remark \ref{rem:gains}. Moreover, we set $\gamma=1.0$, and $L=4.0$ satisfying Assumption \ref{as:bound}.  For the first experiment, the triggering threshold is fixed at {$\delta_{ij}(t)=\delta=0.02$}.  All constrained optimization programs required in the procedures of Remarks~\ref{rem:gains} and~\ref{rem:event} were implemented using the \texttt{minimize} routine from the \texttt{scipy} package.

Figure~\ref{fig:sim} summarizes the results. The top panel shows the estimated derivatives $\hat{s}_{i,1}(t)$ together with the true average derivative $\dot{\bar s}(t)$ for a trigger threshold $\delta=0.02$. All agents correctly recover the average trajectory with small steady-state deviations induced by the trigger. The middle panel depicts the absolute errors $|\hat{s}_{i,1}(t)-\dot{\bar s}(t)|$, which converge to a neighborhood of the origin whose size scales with $\sqrt{\delta}$, in agreement with Theorem~\ref{prop:trigger}. The theoretical bound $c_1\sqrt{\delta}$ with $c_1=7.9$ is also shown, where it can be observed that the steady state error never exceeds the guaranteed bound. The bottom-left panel reports the maximum steady-state error as a function of $\delta$, obtained by sweeping $\delta\in[0,0.14]$ across 100 experiments and measuring the error after $0.8T$ seconds. The numerical results never exceed the guaranteed predicted $c_1\sqrt{\delta}$ bound, where some expected conservativeness is observed. Finally, the bottom-right panel shows the corresponding fraction of events relative to continuous transmission, illustrating that larger values of $\delta$ substantially reduce communication load.

{Now, we repeat the experiment with other time varying and state dependent triggering alternatives. Figure~\ref{fig:trigger} shows the behavior of the event-triggered distributed differentiator under three representative choices of the threshold \(\delta_{ij}(t)\). As before, with a constant threshold \(\delta_{ij}(t)=\delta\) (left column), the algorithm exhibits practical convergence. After a short transient, the tracking error remains bounded within a neighborhood whose size is determined by \(\delta\), while inter-event times stay strictly positive and do not show any long term degradation. When a vanishing threshold \(\delta_{ij}(t)=\delta \exp(-t/2)\) is employed (middle column), the tracking error decays asymptotically to zero, confirming asymptotic exactness of the distributed differentiator. However, this improvement in accuracy is accompanied by progressively smaller inter-event times, reflecting the increasing communication effort required to maintain exact tracking under persistently varying signals. Finally, the state dependent trigger \(\delta_{ij}(t)=\delta+\sigma|\hat{s}_{i,0}(\tau_k)-\hat{s}_{j,0}(\tau_k)|\) (right column) provides an intermediate behavior. During the transient, larger thresholds reduce communication when agents are far from agreement, while in steady state the threshold effectively saturates at \(\delta\), preventing arbitrarily fast triggering. As a result, inter-event times remain uniformly bounded away from zero while the steady state tracking error remains small and tunable through \(\delta\), illustrating the intended tradeoff between communication sparsity and terminal accuracy.}

\begin{figure}[t]
    \centering
    \includegraphics[width=\linewidth]{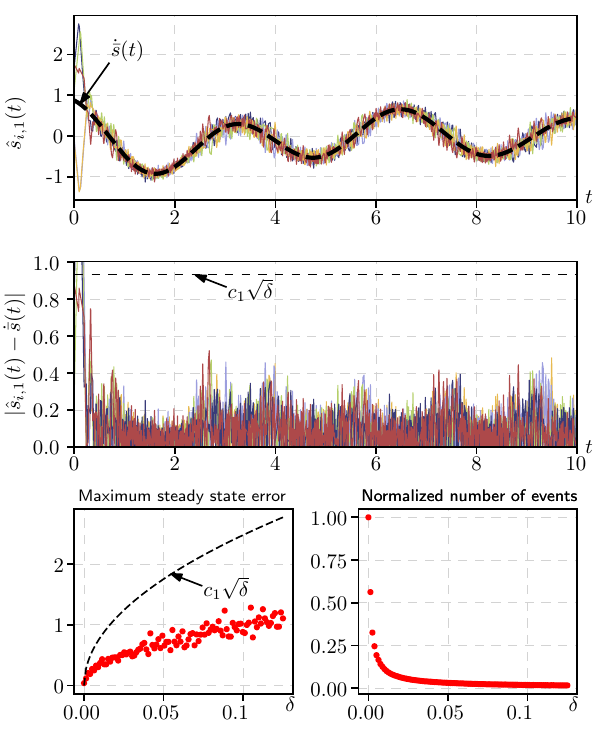}
    \caption{Simulation of the event-triggered proposed protocol. 
    Top: estimated derivatives $\hat{s}_{i,1}(t)$ (solid) and true average derivative $\dot{\bar{s}}(t)$ (black dashed) for an event-trigger with threshold $\delta=0.02$. 
    Middle: absolute errors $|\hat{s}_{i,1}(t)-\dot{\bar{s}}(t)|$. The theoretical steady state error bound $c_1\sqrt{\delta}$ with $c_1=7.9, \delta=0.02$ is shown (black dashed).
    Bottom left: steady-state error versus triggering threshold $\delta$ (red dots) with theoretical $c_1\sqrt{\delta}$ bound (black dashed). {Here, $\delta=0$ corresponds to the performance of the standard REDCHO \cite{redcho}.} 
    Bottom right: fraction of the total number of events relative to the case of full transmission.}
    \label{fig:sim}
\end{figure}

\begin{figure*}[t]
    \centering
    \includegraphics[width=\linewidth]{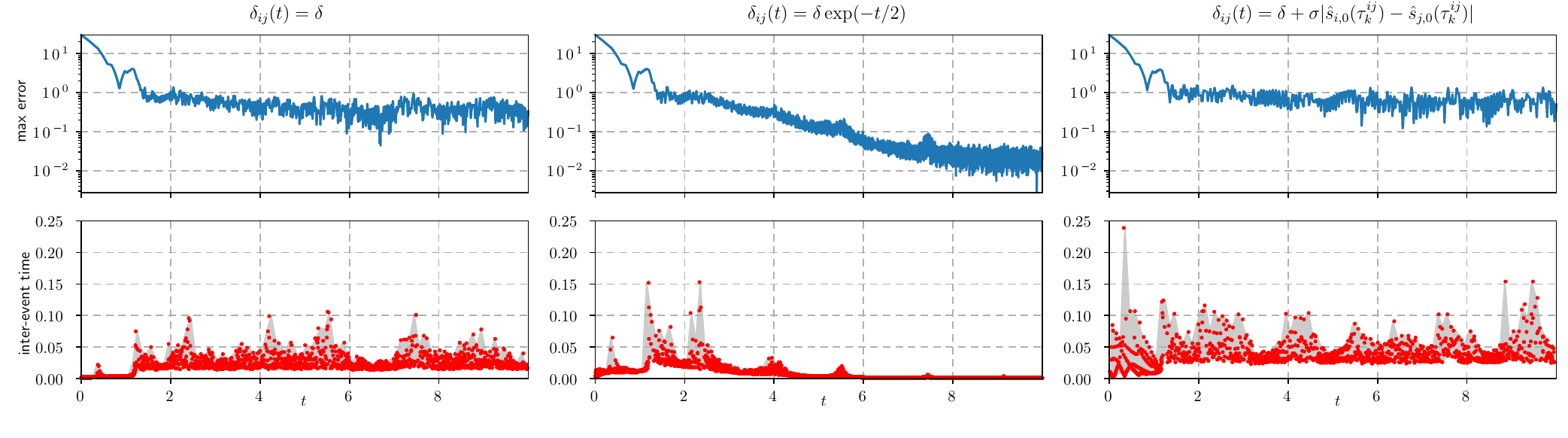}
    \caption{{Comparison of three event thresholds $\delta_{ij}(t)$ in the event-triggered distributed differentiator. Left column: constant threshold $\delta_{ij}(t)=\delta$, yielding practical convergence to a steady state error floor while preserving strictly positive inter-event times. Middle column: vanishing threshold $\delta_{ij}(t)=\delta\exp(-t/2)$, yielding asymptotically exact tracking at the cost of progressively smaller inter-event times. Right column: state dependent threshold $\delta_{ij}(t)=\delta+\sigma|\hat{s}_{i,0}(\tau_{k}^{ij})-\hat{s}_{j,0}(\tau_{k}^{ij})|, \sigma=0.15$, yielding a favorable compromise in which inter-event times remain bounded away from zero along the run, while the steady state error remains small and tunable through $\delta$. Top row: representative maximum tracking error over all agents for each time $t$. Bottom row: inter-event times along the same experiment, with the envelope between maximum and minimum in gray.}}
    \label{fig:trigger}
\end{figure*}
\begin{rem}
\label{rem:conservativeness}
    {As mentioned before, the gains used in the experiments are set to \(k_0>3.45\) and \(k_1>12.4\), thus complying with the sufficient conditions in \eqref{eq:gains}. A natural concern is the conservativeness of these bounds. In this respect, additional simulations indicate that significantly smaller gains, approximately \(k_0 \approx 1.25\) and \(k_1 \approx 2.25\), still yield stable behavior for the considered scenario. This observation suggests that the gain conditions derived from the Lyapunov analysis are conservative, as expected for sufficient conditions providing global and rigorous guarantees. Developing less conservative design rules while preserving formal guarantees is not trivial and is left as future work.}
\end{rem}

\section{Conclusions}
We presented a Lyapunov framework for distributed differentiation that extends existing high order sliding-mode consensus schemes. By isolating the structural features shared with the super-twisting algorithm and encoding them into an abstract model, we derived explicit gain conditions and established global finite-time convergence. This analysis resolves the lack of systematic tuning guidelines and strengthens previous results that only provided local stability. {We also developed an event-triggered implementation based on constant, time-varying and state dependent event-trigger thresholds. For these options, we ruled out the existence of Zeno behavior and obtained error bounds that scale with the triggering threshold, thereby quantifying the accuracy–communication trade-off in the networked system.} The proposed framework thus provides a principled basis for the design of distributed differentiators, with rigorous guarantees in both continuous and event-triggered settings.

\appendix

\section{Auxiliary Results}

This section collects auxiliary concepts and results on homogeneity \citep{bernuau2014} and convex analysis \citep{rockafellar1970}.  
A function $V:\mathbb{R}^n\to\mathbb{R}$ is called \emph{homogeneous of degree $d$} if  
$$
V(\lambda \mf{x}) = \lambda^d V(\mf{x}), \qquad \forall \lambda>0,~ \mf{x}\in\mathbb{R}^n.
$$

This notion extends naturally to \emph{weighted homogeneity}.  
Given weights $\mf{r}=[r_1,\dots,r_n]$, define the dilation  
$$
\Delta_\lambda = \mathrm{diag}(\lambda^{r_1},\dots,\lambda^{r_n}).
$$
A function $V:\mathbb{R}^n\to\mathbb{R}$ is said to be \emph{$\mf{r}$-homogeneous of degree $d$} if  
$$
V(\Delta_\lambda \mf{x}) = \lambda^d V(\mf{x}), \qquad \forall \lambda>0,~ \mf{x}\in\mathbb{R}^n.
$$

\begin{lem}[\cite{Cruz2019}, Lemma~5]
\label{le:comparison}
Let $V_1,V_2:\mathbb{R}^n\to\mathbb{R}$ be continuous functions, homogeneous with respect to the dilation 
$
\Delta_\lambda = \mathrm{diag}(\lambda^{r_1},\dots,\lambda^{r_n}),
$ 
of degrees $m_1,m_2>0$, respectively.  
Assume $V_1$ is positive definite.  
Then, for every $\mf{x}\in\mathbb{R}^n$,
$$
\underline{v}V_1(\mf{x})^{\frac{m_2}{m_1}}
\;\leq\; V_2(\mf{x}) 
\;\leq\; \overline{v}V_1(\mf{x})^{\frac{m_2}{m_1}},
$$
where
$$
\underline{v} = \inf_{\mf{x}:V_1(\mf{x})=1} V_2(\mf{x}), 
\qquad 
\overline{v} = \sup_{\mf{x}:V_1(\mf{x})=1} V_2(\mf{x}).
$$
\end{lem}

\begin{lem}[\cite{Cruz2019}, Lemma~4]
\label{le:gain:ineq}
Let $\Gamma,\Pi:\mathbb{R}^n\to\mathbb{R}$ be continuous functions, homogeneous with respect to the dilation
$
\Delta_\lambda = \mathrm{diag}(\lambda^{r_1},\dots,\lambda^{r_n}), \mf{r}=[r_1,\dots,r_n],
$
of common degree $m>0$, with $\Gamma(\mf{x})\ge 0$.  
Suppose that
$$
\{\mf{x}\in\mathbb{R}^n\setminus\{\mf{0}\} : \Gamma(\mf{x})=0\} 
~\subseteq~ \{\mf{x}\in\mathbb{R}^n\setminus\{\mf{0}\} : \Pi(\mf{x})<0\}.
$$
Then, for any $c>0$ there exists $k>0$ such that
$$
-k\Gamma(\mf{x}) + \Pi(\mf{x}) \;\le\; -c\|\mf{x}\|_\mf{r}^m,
\qquad \forall \mf{x}\in\mathbb{R}^n.
$$
In particular, the following condition suffices
$$
k \;\ge\; \sup_{\|\mf{x}\|_\mf{r}=1,\Gamma(\mf{x})>0} 
\frac{\Pi(\mf{x})+c}{\Gamma(\mf{x})}.
$$
\end{lem}

\begin{lem}
\label{le:properties}
Let Assumption~\ref{as:basic} hold. Then the following statements are true:
\begin{enumerate}[label=\roman*)]
    \item \label{prop:fenchel}  
    $U(\mf{x}_0) + U^*(\mf{x}_1) - \langle \mf{x}_0,\mf{x}_1\rangle \ge 0$,  
    with equality if and only if $\mf{x}_1 = \nabla U(\mf{x}_0)$.
    
    \item \label{prop:strict_convex} $U^*$ is strictly convex.
    
    \item \label{prop:pd} $U^*$ is positive definite.
    
    \item \label{prop:hom} $U^*$ is homogeneous of degree $3$.
    
    \item \label{prop:euler} 
    $\langle \nabla U(\mf{x}_0), \mf{x}_0\rangle \ge U(\mf{x}_0)$,  
    $\langle \nabla U^*(\mf{x}_1), \mf{x}_1\rangle \ge U^*(\mf{x}_1)$.
    
    \item \label{prop:dilation} As functions of $(\mf{x}_0,\mf{x}_1)$, the terms $U(\mf{x}_0)$, $U^*(\mf{x}_1)$, and $\langle \mf{x}_0,\mf{x}_1\rangle$ are all $\mf{r}$-homogeneous of degree $3$ with respect to the dilation  
    $
    \Delta_\lambda = \mathrm{diag}(\lambda^{r_1},\dots,\lambda^{r_n}), \mf{r}=[2\mathds{1}^\top,\mathds{1}^\top].
    $
    
    \item \label{prop:gamma} For any $\beta \ge 7$,  
    $$
    U(\mf{x}_0) + (1+\beta)U^*(\mf{x}_1) - 2\langle \mf{x}_0,\mf{x}_1\rangle \;\ge\; 0,
    \quad \forall \mf{x}_0,\mf{x}_1\in\mathcal{X}.
    $$
\end{enumerate}
\end{lem}

\begin{pf}
Item~\ref{prop:fenchel} is the standard Fenchel--Young inequality, which holds by the construction of $U^*$ in \eqref{eq:dual}.  Item~\ref{prop:strict_convex} follows from standard results in convex analysis on the dual \cite[Section 3.3.1]{boyd_convex}. For Item~\ref{prop:pd}, recall that Item~\ref{prop:fenchel} holds for all $\mf{x}_0,\mf{x}_1$.  
Setting $\mf{x}_0=\mf{0}$ yields
$$
U^*(\mf{x}_1) \;\ge\; \langle \mf{0},\mf{x}_1\rangle - U(\mf{0}) = 0.
$$
Moreover,
$$
U^*(\mf{0}) = \sup_{\mf{x}_0}\{-U(\mf{x}_0)\} = -\inf_{\mf{x}_0} U(\mf{x}_0) = 0,
$$
since $U$ is positive definite. Hence $U^*$ is also positive definite.  For Item~\ref{prop:hom}, observe that
$$
\begin{aligned}
U^*(\lambda \mf{x}_1) 
&= \sup_{\mf{x}_0\in\mathcal{X}}\{\langle \mf{x}_0,\lambda \mf{x}_1\rangle - U(\mf{x}_0)\} \\
&= \sup_{\mf{x}_0'\in\mathcal{X}}\{\langle \lambda^2 \mf{x}_0', \lambda \mf{x}_1\rangle - U(\lambda^2 \mf{x}_0')\} \\
&= \sup_{\mf{x}_0'\in\mathcal{X}}\{\lambda^3\langle \mf{x}_0', \mf{x}_1\rangle - \lambda^3 U(\mf{x}_0')\} \\
&= \lambda^3 U^*(\mf{x}_1),
\end{aligned}
$$
where the change of variables $\mf{x}_0=\lambda^2 \mf{x}_0'$ was used, together with the $3/2$-homogeneity of $U$.  For Item~\ref{prop:euler}, let $f:\mathcal{X}\to\mathbb{R}$ be any differentiable positive definite convex function (e.g., $U$ or $U^*$). Convexity yields
$$
f(\mf{x}') \;\ge\; f(\mf{x}) + \langle \nabla f(\mf{x}), \mf{x}'-\mf{x}\rangle.
$$
Setting $\mf{x}'=\mf{0}$ gives
$$
0=f(\mf{0}) \;\ge\; f(\mf{x}) - \langle \nabla f(\mf{x}), \mf{x}\rangle,
$$
which implies $\langle \nabla f(\mf{x}), \mf{x}\rangle \ge f(\mf{x})$.  
Applying this relation to $U$ and $U^*$ gives the claim. For Item~\ref{prop:dilation}, note that
$$
U(\lambda^2 \mf{x}_0) = (\lambda^2)^{3/2} U(\mf{x}_0) = \lambda^3 U(\mf{x}_0),
$$
$$
\langle \lambda^2 \mf{x}_0, \lambda \mf{x}_1\rangle = \lambda^3 \langle \mf{x}_0, \mf{x}_1\rangle,
$$
and $U^*$ is homogeneous of degree $3$ by Item~\ref{prop:hom}. Thus each term is $\mf{r}$-homogeneous of degree $3$ under $\mf{r}=[2\mathds{1}^\top,\mathds{1}^\top]$. Finally, for Item~\ref{prop:gamma}, the Fenchel--Young inequality implies
$$
-\langle \mf{x}_0, 2\mf{x}_1\rangle \;\ge\; -U^*(2\mf{x}_1) - U(\mf{x}_0).
$$
Therefore,
$$
\begin{aligned}
&U(\mf{x}_0) + (1+\beta)U^*(\mf{x}_1) - 2\langle \mf{x}_0,\mf{x}_1\rangle \\
&\ge (1+\beta)U^*(\mf{x}_1) - U^*(2\mf{x}_1) \\
&= (1+\beta)U^*(\mf{x}_1) - 2^3 U^*(\mf{x}_1) \\
&= (\beta-7) U^*(\mf{x}_1) \;\ge\; 0,
\end{aligned}
$$
for all $\beta\ge 7$.  
\qed
\end{pf}

\end{document}